\documentclass[dvipsnames]{IEEEtran}
\usepackage{trdean}
\usepackage{flushend}
\usepackage{xcolor}
\usepackage{soul}	

\usepackage{algorithm,algorithmic}
\usepackage{setspace}

\newcommand{\refprog}{(\ref{eq:obj})--(\ref{eq:const})\,\,}

\newcommand{\refprogn}{(\ref{eq:obj})--(\ref{eq:const})}

\soulregister{\ref}{1}
\soulregister{\cite}{1}

\newcommand{\change}[1]{#1}
\newcommand{\remove}[1]{}

\title{Fast Blind MIMO Decoding \\through Vertex Hopping}
\author{Thomas~R.~Dean,~\IEEEmembership{Member,~IEEE,} 
Jonathan~R.~Perlstein,\\
Mary~Wootters,~\IEEEmembership{Member,~IEEE,} and 
Andrea~J.~Goldsmith,~\IEEEmembership{Fellow,~IEEE}%
\thanks{T. Dean and A. Goldsmith are with the Department of Electrical Engineering, Stanford, CA
94305 USA (e-mail: trdean@stanford.edu, andrea@ee.stanford.edu).}
\thanks{M. Wootters is with the Department of Computer Science and the Department of
Electrical Engineering, Stanford, CA 94305 USA (e-mail: marykw@stanford.edu).}
\thanks{J. Perlstein is with the Department of Computer Science, Stanford, CA
94305 USA (email: jrperl@cs.stanford.edu).}
\thanks{This work was presented in part at the 2018 Asilomar Conference on Signals, Systems, and Computers in Pacific Grove, CA \cite{asilomar}.
T. Dean is supported by the Fannie and John Hertz Foundation. This work was supported in part by the NSF Center for Science of Information under Grant CCF-0939370}}
\begin{document}
\maketitle 
\begin{abstract}
We present an algorithm that efficiently performs blind decoding of MIMO signals.  
That is, given no channel state information (CSI) at either the transmitter or receiver, our algorithm takes a block of samples and returns an estimate of the underlying data symbols.  
In prior work, the problem of blind decoding was formulated as a non-convex optimization problem.
In this work, we present an algorithm that efficiently solves this 
non-convex problem in practical settings.
This algorithm leverages concepts of linear and mixed-integer linear programming.  Empirically, we show that our technique has an error performance close to that of zero-forcing with perfect CSI at the receiver.  Initial estimates of the runtime of the algorithm presented in this work suggest that the real-time blind decoding of MIMO signals is possible for even modest-sized MIMO systems.

\vspace{1mm}
\end{abstract}

\section{Introduction}
In this work we propose an efficient method to blindly estimate MIMO channels
and decode the underlying transmissions. 
A previous work, \cite{dean2018blind}, has shown that 
as long as the channel gain matrix is non-singular, 
the geometry of the constellation can be exploited to recover the underlying data up to some small amount of remaining ambiguity.  
In \cite{dean2018blind}, the authors formulate the problem of blind MIMO decoding as a non-convex optimization problem and provide theoretical guarantees as to when the interior point algorithm, with a logarithmic barrier function, correctly solves this problem.

As the size of the MIMO system grows, interior-point based algorithms given in \cite{dean2018blind} become ineffective  due to both an increasing proportion of spurious optima as well as numerical instability.  In this work, we propose an algorithm inspired by techniques commonly used to solve linear and mixed-integer programs that allow us to perform blind decoding on higher order MIMO systems.  More importantly, this approach is far more computationally efficient than the approach in \cite{dean2018blind}, such that real-time decoding is possible \change{for systems as large as $n=8$.}

A blind decoding algorithm that is realistic in terms of both sample size requirements and computational complexity has important applications.  
Current trends in wireless system design are moving towards systems with shorter wavelengths, higher user mobility and more antennas per user \cite{andrews2014will, boccardi2014five, rappaport2013millimeter}.
Thus, channel state information (CSI) is increasingly rapidly varying and difficult to acquire. 
Additionally, new systems being proposed demand increased reliability and decreased latency \cite{osseiran2014scenarios}.
These factors combined imply that reducing resource overhead from channel estimation will become even more important in the development of future wireless systems.

An efficient blind decoding algorithm can enable high-rate communications in environments where the channel changes too rapidly to be measured or where communication occurs in short bursts.  Capacity of channels with CSI unknown at the receiver have been considered extensively in the literature (e.g. \cite{telatar1999capacity,goldsmith2003capacity,zheng2002communication} and references therein).
While it is theoretically possible to achieve reliable, high-rate communications in these conditions, few schemes exist that do so practically.   



\change{
    A number of previous works have considered blind decoding of MIMO systems;
    we briefly differentiate our work from prior works.  
    Many statistical approaches have been considered for blind MIMO decoding, such as {\cite{talwar1996blind,hansen1997hyperplane}},
    which are generally based on covariance matrix estimation.
    These approaches require prohibitively large sample sizes especially when channels are not well conditioned.
    Other techniques such as {\cite{ling2015blind}} require the mixing process to be structured.
    A different set of works, for example
    {\cite{zhang2018blind,mezghani2018blind, ghavami2018blind}}, consider blind decoding of massive MIMO systems by exploiting channel sparsity; in contrast our approach targets smaller (roughly $n < 12$) system sizes and dense scattering environments.
    See {\cite{dean2018blind}} for a more complete comparison of our approach to prior works.}

In this work, we compare the performance of our new approach only to \cite{dean2018blind}.  This is because previous approaches have prohibitively large sample size requirements, often growing exponentially in the number of transmit antennas.  To our knowledge, \cite{dean2018blind} is the only blind MIMO decoding method that has sample size requirements that are less than the block length of modern wireless systems.

We now highlight several notable features of the algorithm presented in this work: 
\begin{itemize}
\item For an $n \times n$ MIMO system we can typically decode a block of $k$ channel uses in $O\left( n^4 k \right)$ operations.  We also characterize the scenarios where the number of operations for decoding exceeds this bound.  For comparison, when CSI is known perfectly at the receiver, efficient MIMO decoding algorithms such as zero-forcing and MMSE require $O\left( n^3 k \right)$ operations to decode $k$ channel uses.
\item In the limit of high SNR, given appropriate inputs described in {\cite{dean2018blind}}, our approach solves the blind decoding problem with a success rate approaching 1 for systems as large as $n=12$.  In comparison, the success rate of the approach in {\cite{dean2018blind}} was bounded below 1 beyond $n=5$ and negligibly small beyond $n=8$.
\item We implement the proposed algorithm in the Rust programming language and show that the runtime of our algorithm is several orders of magnitude faster than the approach given in \cite{dean2018blind}.  For $n\leq 8$ our implementation is approaching fast enough runtimes to enable real-time blind decoding of data streams without specialized hardware.  We note that $n \leq 8$ captures nearly all MIMO systems in use today \cite{3gppRI, 80211RI}.
\item At low SNR, our approach has BER performance nearly matching zero-forcing with perfect CSI.  At high SNR, our technique has an increased BER compared to techniques with perfect CSI but the BER vanishes in the limit of no noise.  In all cases, we outperform maximum-likelihood decoding when the CSI estimate at the receiver has as little as 1\% estimation error.
\end{itemize}

In Table {\ref{tab:speed}}, we report the runtime of an implementation of our algorithm written in the Rust programming language, run on a single core of an Intel i7 2.2GHz processor.  
These numbers are preliminary as there are many more optimizations remaining to be implemented in our solver.  
For comparison, we also report the runtime of the gradient descent algorithm of \cite{dean2018blind}, implemented using MATLAB's \texttt{fmincon} solver.
In addition, we report that probability that both algorithms return a correct solution to the blind decoding problem.  
For all values of $n$, the approach presented in this paper is both faster and more reliable.  
A more detailed discussion of both the runtime and success probability is presented in Section {\ref{sec:new_alg}}.

The remainder of this paper is organized as follows.  Section \ref{sec:model} describes the system model and notation used throughout this work.  Section \ref{sec:old_alg} revisits several relevant theoretical results presented in \cite{dean2018blind}.  Section \ref{sec:new_alg} provides a high-level description of our new algorithm along with empirical results. Sections \ref{sub:initial} and \ref{sec:vertex} describe the core components of our algorithm.  Section \ref{sub:noise} discusses our algorithm in the presence of AWGN. Conclusions are offered in Section \ref{sec:conclusion}.  

\blockcomment{
\begin{table}
\caption{Runtime on a single core of a 2.2GHz Intel i7}
\label{tab:speed}
\centering
\scalebox{0.85}{
    \begin{tabular}{|c|c||c|c||c|c|}
    \hline
     \multicolumn{2}{|c||}{} & \multicolumn{2}{c||}{Vertex Hopping} & \multicolumn{2}{c|}{Gradient Descent \cite{dean2018blind}}\\
    \hline
    n & k & Pr. Success & Time (s) & Pr. Success & Time (s) \\
    \hline
    2 & 8 & 1.00 & 1.83e-5 & 0.99 & 3.01e-2 \\
    3 & 13 & 1.00 & 6.46e-5 & 0.99 & 6.33e-2 \\
    4 & 18 & 1.00 & 1.75e-4 & 0.99 & 0.13 \\
    5 & 18 & 1.00 & 2.50e-4 & 0.97 & 0.30 \\
    6 & 22 & 1.00 & 8.03e-4 & 0.93 & 0.59 \\
    8 & 30 & 0.99 & 3.52e-3 & 0.80 & 3.5 \\
    10 & 100 & 0.99 & 1.91e-2 & 0 & - \\
    12 & 144 & 0.99 & 2.28e-1 & 0 & - \\
    \hline
    \end{tabular}
}
\end{table}
}

\begin{table}
\caption{Runtime of the vertex hopping algorithm on a single core of a 2.2GHz Intel i7}
\label{tab:speed}
\vspace{2mm}

\centering
\ra{1.3}
\scalebox{0.85}{
    \begin{tabular}{@{}llllllllll@{}}
    \toprule
     \multicolumn{2}{c}{} &&& \multicolumn{2}{c}{Vertex Hopping} &&& \multicolumn{2}{c}{Gradient Descent \cite{dean2018blind}}\\
     \cmidrule{5-6} \cmidrule{9-10}
    n & k &&& Pr. Success & Time (s) &&& Pr. Success & Time (s) \\
    \midrule
    2 & 8 &&& 1.00 & 1.83e-5 &&& 0.99 & 3.01e-2 \\
    3 & 13 &&& 1.00 & 6.46e-5 &&& 0.99 & 6.33e-2 \\
    4 & 18 &&& 1.00 & 1.75e-4 &&& 0.99 & 0.13 \\
    5 & 18 &&& 1.00 & 2.50e-4 &&& 0.97 & 0.30 \\
    6 & 22 &&& 1.00 & 8.03e-4 &&& 0.93 & 0.59 \\
    8 & 30 &&& 0.99 & 3.52e-3 &&& 0.80 & 3.5 \\
    10 & 100 &&& 0.99 & 1.91e-2 &&& 0 & - \\
    12 & 144 &&& 0.99 & 2.28e-1 &&& 0 & - \\
    \bottomrule
    \end{tabular}
}
\end{table}

\section{System Model and Notation}
\label{sec:model}
In this work we consider $n \times n$ real-valued channel gain matrices, denoted $\mathbf{A}$.  \change{Unless otherwise specified}, we draw $\mathbf{A}$ where each entry is i.i.d. and normally distributed with zero mean and unit variance, $\mathcal{N}(0,1)$. Our technique requires only that $\mathbf{A}$ be full rank and thus $\mathbf{A}$ may be drawn from an arbitrary distribution. 
We assume AWGN, drawn from $\mathcal{N}\left(0,\sigma^2\right)$, is present in the channel.
Finally, we consider a block fading model where the channel gain matrix is constant over a period of $k$ channel uses after which it is redrawn independently.  This work focuses on the transmission of BPSK signals over such channels.\footnote{While we only consider $n \times n$ real channels, we note that the results in this work can be extended to $n \times n$ complex-valued channels by considering the usual $2n \times 2n$ equivalent real-valued channel gain matrix and can be extended  rectangular channels as discussed in \cite{dean2018blind}.
The results in \cite{dean2018blind} extend to general MPAM constellations; the performance of the algorithms presented here under the presence of higher-order modulation is a topic of on going research.}  

The vector $\mathbf{x} \in \{-1,+1\}^n$ denotes the symbols transmitted in a single channel use and the matrix $\mathbf{X} \in \{-1,+1\}^{n \times k}$ denotes the set of symbols transmitted over a single block of $k$ channel uses. Likewise, single observations and blocks of symbols at the receiver are denoted as $\mathbf{y}$ and $\mathbf{Y}$ respectively.

We assume that \emph{no CSI} is available at either the transmitter or the receiver.  Without the aid of pilot symbols or any knowledge of the underlying data symbols, the receiver attempts to recover an estimate of $\mathbf{X}$, denoted $\hat{\mathbf{X}}$.  However, as discussed in \cite{dean2018blind}, without additional side information, the receiver is only capable of recovering $\mathbf{X}$ up to an acceptable transform matrix (ATM), meaning that within each block, $\hat{\mathbf{X}}$ is correct up to permutation and negation of the rows.  We denote the set of ATMs as $\mathcal{T}$.
In the high SNR limit, we say that an algorithm solves the blind decoding problem if, given only $\mathbf{Y}$ as input, it returns a value of $\hat{\mathbf{X}}$ that is equivalent to $\mathbf{X}$ up to an ATM.

The $i$th column of the matrix $\mathbf{A}$ is denoted as $\mathbf{a}_i$ and its
$j$th row as $\mathbf{a}^{(j)}$. The vector $\mathbf{e}_i$ represents the $i$th
element of the standard basis. The set $\text{cols}(\mathbf{X})$ denotes the set
of vectors that comprise the columns of $\mathbf{X}$ and
$\text{vec}(\mathbf{X})$ denotes the $nk \times 1$ vector that consists of the
entries of $\mathbf{X}$ in row-major order. 
The notation $\mathbf{U}^{-\intercal}$ is shorthand for $\left( \mathbf{U}^{-1} \right)^\intercal$.  The symbol $\mathbb{R}_+$ denotes the domain of non-negative real numbers, whereas $\mathbb{R}_{++}$ denotes the positive real numbers.

\section{Fitting a Parallelepiped}
\label{sec:old_alg}
In \cite{dean2018blind} the authors formulate the blind decoding problem as the following non-convex optimization problem:
\begin{align}
& \underset{\mathbf{U}}{\text{maximize}} 
& & \log | \det \mathbf{U} | \label{eq:obj}\\
& \text{subject to}
& & \| \mathbf{Uy}_i \|_\infty \leq 1 + c \cdot \sigma, \; i = 1, \ldots, k,
\label{eq:const}
\end{align}
where $c$ is some margin which is chosen based on the noise variance.
Given proper input, the set of optimal $\mathbf{U}$ are equivalent to $\mathbf{A}^{-1}$ up to an ATM.
Geometrically, this problem can be interpreted as fitting the minimum volume parallelepiped that matches the observed samples. 
In \cite{dean2018blind} the authors show that despite the fact that the problem is not convex, under certain assumptions, gradient descent returns the correct solution with high probability.  
Here, we briefly recount several important theoretical facts proven in \cite{dean2018blind} about the problem given by \refprogn. 
We initially focus on the noiseless case ($\sigma=0$), and return to the case $\sigma \neq 0$ in Section \ref{sub:noise}.  We also assume that $\mathbf{Y}$ is full rank; if it is not then \refprog is not a well-posed problem.

In \pref{eq:const}, each $\mathbf{y}_i$ imposes two linear constraints on each row $\mathbf{u}^{(j)}$,
that is $-1 \leq \left\langle \mathbf{u}^{(j)}, \mathbf{y}_i \right\rangle \leq 1$, for all $i,j$.  The feasible region is thus an $n^2$-dimensional polytope. 
We say that a given $\mathbf{U}$ is at a vertex of this polytope if $\mathbf{UY}\in \{-1,+1\}^{n \times k}$.  Note that the objective function is not defined at all vertices of the feasible region; if $\mathbf{UY}$ is not full rank, this implies $\mathbf{U}$ is singular and the value of \pref{eq:obj} is not defined.
If two vertices $\mathbf{U}_1$ and $\mathbf{U}_2$ are adjacent (share an edge of the polytope) then this also implies that the Hamming distance between $\mathbf{U}_1 \mathbf{Y}$ and $\mathbf{U}_2 \mathbf{Y}$ is 1.  

A matrix $\mathbf{X} \in [-1, +1]^{n \times k}$, and corresponding set $\text{cols}(\mathbf{X}) \subseteq [-1,+1]^{n}$ with $k \geq n$, has the
\emph{maximal
subset property} (MSP) if there is a subset $\text{cols}(\mathbf{V}) \subseteq \text{cols}(\mathbf{X})$ of size $n$ so that if
$\mathbf{V} \in \mathbb{R}^{n \times n}$ is the matrix with elements of $\text{cols}(\mathbf{V})$ as
columns, then
\begin{equation*}
|\det \mathbf{V}| = \underset{\mathbf{W} \in [-1, +1]^{n \times n}}{\mathrm{max}}
           |\det \mathbf{W}|.
\end{equation*}
Since $\det(\mathbf{V})$ is linear in the columns of $\mathbf{V}$, then $\mathbf{V}$ also maximizes the determinant amongst all matrices in $\{-1, +1\} ^ {n \times n}$.
We note the following additional facts about the program given in \refprogn:

\begin{itemize}
\item If the matrix $\mathbf{X}$ has the maximal subset property, then the set of global optima of \refprog contain all solutions to the blind decoding problem.  This is proven in \cite{dean2018blind}.
\item The gradient of the objective function is given by
\begin{equation}
\label{eq:gradient}
\nabla \left( \log | \det \mathbf{U} | \right) =
\mathbf{U}^{-\intercal}.
\end{equation}
\item For $n<6$, all optima of \refprog are global optima. In \cite{dean2018blind} the authors give specific values of $\mathbf{X}$ that guarantee all optima of \refprog are solutions to the blind decoding problem for $n\leq 4$.  The case $n=5$ is discussed in \cite{perlstein2018n5case}.
\item Solutions to the blind decoding problem lie on vertices of the feasible region.  When $n$ is such that a Hadamard matrix exists, all optima are strict and lie on vertices.  In all cases, optima will only lie on the boundary of the feasible region.  This is proven in \cite{dean2018blind}.
\end{itemize}

While it is shown in \cite{dean2018blind} that the interior point method with a
logarithmic barrier function provably solves the non-convex blind decoding
problem for small $n$, this approach has several practical limitations.  First,
for $n>5$, local optima exist,  and this approach is no longer guaranteed to be correct. As $n$ grows substantially beyond 5, the proportion of optima that are local increases and thus this approach becomes less effective.
Second, this approach is not efficient from a computational perspective.  All solutions are located on vertices of the feasible region.  In general, barrier methods are not well suited to solve this class of problems \cite{boyd_opt}. The gradient of the objective function varies rapidly near the boundary of the feasible region which, along with the large number of linear constraints imposed by \pref{eq:const}, creates numerical instability.  As a result, even at low dimension, off-the-shelf solvers such as MATLAB's \texttt{fmincon} will require a large number of Newton steps to converge.  Further, in our numerical experiments in solving \refprogn, we have observed that such interior point solvers almost never converge for $n>8$.  

\blockcomment{
\begin{table}
\caption{\change{Probability a $\pm 1$-valued matrix drawn uniformly at random has the MSP.}}
\label{tab:sample}
\centering
\scalebox{0.85}{
    \begin{tabular}{|c||c|c|c|}
    \hline
    n & 90\% & 99\% & $1 - 10^{-6}$ \\
    \hline
    2 & 5 & 9 & 22 \\
    4 & 10 & 13 & 26\\
    6 & 14 & 18 & 29 \\
    8 & 16 & 20 & 34 \\
    10 & 28 & 34 & 40 \\
    \hline
    \end{tabular}
}
\end{table}
}

\begin{table}
\caption{Probability a $\pm 1$-valued matrix with $n$ rows drawn uniformly at random has the MSP by number of columns.}
\label{tab:sample}
\vspace{2mm}

\ra{1.3}
\centering
    \begin{tabular}{@{}lllll@{}}
    \toprule
    n\quad && 90\% & 99\% & $1 - 10^{-6}$ \\
    \midrule
    2 && 5 & 9 & 22 \\
    4 && 10 & 13 & 26\\
    6 && 14 & 18 & 29 \\
    8 && 16 & 20 & 34 \\
    10 && 28 & 34 & 40 \\
    \bottomrule
    \end{tabular}
\end{table}

\subsection{\change{Sample Size Requirements}}
\label{sub:samplesize}
\change{
    In {\cite{dean2018blind}}, the authors provide both empirical and analytic results describing the probability that the matrix $\mathbf{X}$ has the MSP.  
    If $\mathbf{X}$ does not have the MSP, then the global optima of ({\ref{eq:obj}})--({\ref{eq:const}}) do not correspond to solutions of the blind decoding problem.  
    In this case, both the approach described in {\cite{dean2018blind}} and the approach described in this work will fail to output a correct solution
    to the blind decoding problem.
    In Table {\ref{tab:sample}}, we provide the probability that a matrix $\mathbf{X}$ drawn uniformly over $\{-1, +1\}^{n \times k}$ has the MSP for several values of $n$ and $k$.  
    Note that the MSP alone is not sufficient to ensure that all global optima of ({\ref{eq:obj}})--({\ref{eq:const}}) are solutions to the blind decoding problem ---
    the criteria for ensuring that no spurious optima exist varies drastically for each value of $n$.
    For most values of $n$, we typically require only one or two columns that are pair-wise independent from the maximal subset of columns that form the MSP.  
    Thus, the probabilities given in Table {\ref{tab:sample}} closely approximate the probability that ({\ref{eq:obj}})--({\ref{eq:const}}) has no spurious optima.
    See {\cite{dean2018blind}} for a more complete discussion of sample size requirements.
}

\section{Algorithm Overview and Performance}
\label{sec:new_alg}
The program \refprog is not linear, nor is it even convex.  One should not necessarily expect tools from convex optimization, let alone linear programming, to work well.  However, we adapt such techniques by leveraging two facts about the problem geometry: the fact that the objective function is multilinear in the rows of $\mathbf{U}$, and that solutions to the blind decoding problem lie on vertices of the feasible region.  Informally, our algorithm attempts to find an appropriate solution to the blind decoding problem by hopping between vertices of the feasible region in a similar manner to the simplex algorithm.  For this reason, we refer to our algorithm as the \emph{vertex hopping algorithm}.  In this section, we give a high-level overview of our approach to blind decoding as well as empirical results demonstrating its BER performance.


The first step of the vertex hopping algorithm is finding a value of $\mathbf{U}$ that is at a non-singular vertex of the feasible region.  We refer to this process as \emph{Vertex Finding}. It turns out that finding such a vertex is a non-trivial task and is often the majority of the work in solving the problem.  This procedure is described in Section \ref{sub:initial}.  

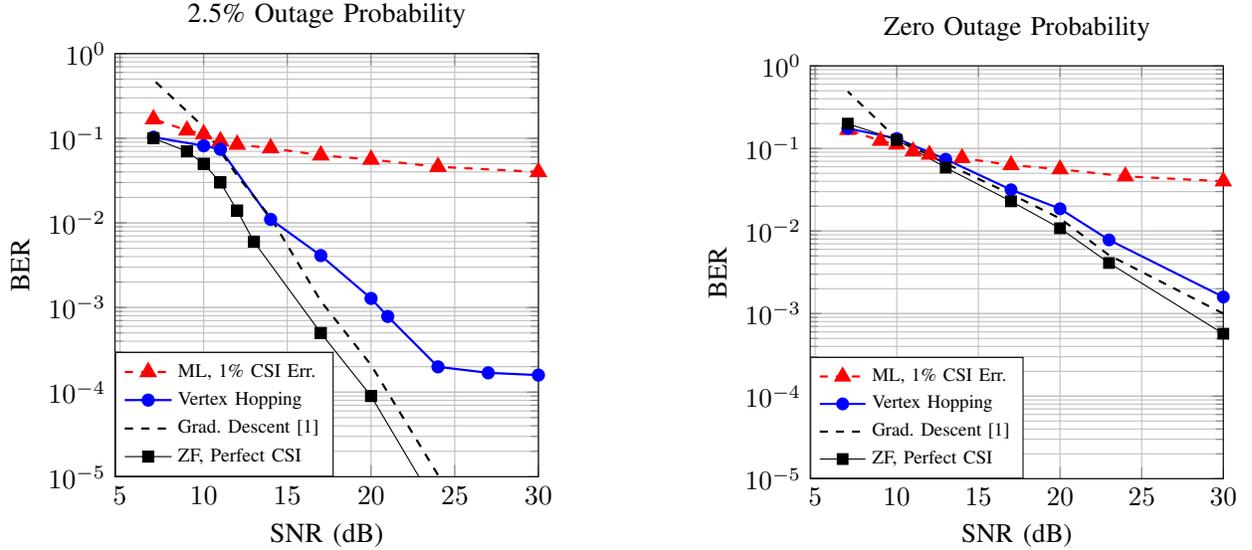
\begin{figure*}
\centering
\hfill
\begin{tikzpicture}
\begin{semilogyaxis}[
    grid=both,
    width=7cm, height=7cm,
    xmax = 30,
    ymin = 1e-5,
    ymax = 1,
    xlabel={SNR (dB)},
    ylabel={BER},
    title={2.5\% Outage Probability},
    legend style={font=\scriptsize,
    	cells={anchor=west},
    	anchor=south west,
        at={(0,0)}}
]

\addplot[red,dashed,thick,mark=triangle*,
mark options={solid,mark size=3}] coordinates
{
(30, 0.040)
(24, 0.046)
(20, 0.056)
(17, 0.0629)
(14, 0.0765)
(12, 0.0842)
(11, 0.0924)
(10, 0.1117)
(9, 0.124)
(7, 0.168)
};

\addplot[blue,thick,mark=*] coordinates
{
(30,1.59e-4)
(27,1.69e-4)
(24,1.99e-4)
(21,7.84e-4)
(20,1.28e-3)
(17,4.11e-3)
(14,1.1e-2)
(11,7.4e-2)
(10,8.15e-2)
(7,1.03e-1)
};

\addplot[black,dashed,thick] coordinates
{
(24, 1.05e-5)
(20, 0.000209375)
(17, 0.001190625)
(14, 0.011578125)
(12, 0.036234375)
(10, 0.1354375)
(7, 0.49)
};

\addplot[black,mark=square*,mark options=
{fill=black,mark size=2}] coordinates
{
(23, 9e-6)
(20,0.00009)
(17,0.0005)
(13,0.006)
(12,0.014)
(11,0.03)
(10,0.05)
(9,0.07)
(7,0.1)
};

\legend{ML, 1\% CSI Err.\\ Vertex Hopping\\ Grad. Descent [1]\\ ZF, Perfect CSI\\}
\end{semilogyaxis}
\end{tikzpicture}
\hspace{1.6cm}
\begin{tikzpicture}
\begin{semilogyaxis}[
    grid=both,
    width=7cm, height=7cm,
    xmax = 30,
    ymin = 1e-5,
    ymax = 1,
    xlabel={SNR (dB)},
    ylabel={BER},
    title={Zero Outage Probability},
    legend style={font=\scriptsize,
    	cells={anchor=west},
    	anchor=south west,
        at={(0,0)}}
]

\addplot[red,dashed,thick,mark=triangle*,
mark options={solid,mark size=3}] coordinates
{
(30, 0.040)
(24, 0.046)
(20, 0.056)
(17, 0.0629)
(14, 0.0765)
(12, 0.0842)
(11, 0.0924)
(10, 0.1117)
(9, 0.124)
(7, 0.168)
};

\addplot[blue,thick,mark=*] coordinates
{
(30,1.59e-3)
(23,7.8e-3)
(20,1.86e-2)
(17,3.16e-2)
(13,7.4e-2)
(10,0.132)
(7,0.175)
};

\addplot[black,dashed,thick] coordinates
{
(30, 0.001)
(23, 0.0051)
(20, 0.0142)
(17, 0.0279)
(13, 0.0657)
(10, 0.1311)
(7, 0.49)
};

\addplot[black,mark=square*,mark options=
{fill=black,mark size=2}] coordinates
{
(30, 5.7e-4)
(23, 0.0041)
(20, 0.0108)
(17,0.0229)
(13,0.0586)
(10,0.1269)
(7,0.2)
};

\legend{ML, 1\% CSI Err.\\ Vertex Hopping\\ Grad. Descent [1]\\ ZF, Perfect CSI\\}
\end{semilogyaxis}
\end{tikzpicture}
\hspace{0.6cm}
\caption{The BER using various MIMO decoding schemes to estimate $\mathbf{X}$ in the presence of AWGN.  Here $n=4, k=30$, where $\mathbf{A}$ is drawn with i.i.d. Gaussian entries. For large values of $n$, both gradient descent and ML decoding become prohibitively computationally expensive.}
\label{fig:awgn}
\end{figure*}

Once we have found an initial non-singular vertex of the feasible region, we efficiently explore neighboring vertices of the feasible region in search of a global optimum.  This is accomplished in a similar manner as the simplex algorithm.  We form a tableau from the linear constraints that define the feasible region and hop between vertices by performing Gauss-Jordan pivots on this data structure.  
This procedure is described in greater depth in Section \ref{sec:vertex}.

At each step, we hop to the neighboring vertex that has the largest increase in objective function and backtrack if we find a local optimum.  We discuss how to identify when we are at a local versus a global optimum in Section \ref{sub:stop}.  Local optima are rare and do not exist for $n \leq 5$.  In nearly every case the algorithm terminates after at most a few hops; we elaborate on when this is not the case in Section \ref{sec:vertex}.

\subsection{Numerical Results}
\label{sec:numerical}
In Figure {\ref{fig:awgn}}, we show empirical results regarding the resulting BER of our scheme in the presence of AWGN
\change{ 
    both with a 2.5\% outage probability and a zero-outage probability for the Gaussian model described in Section {\ref{sec:model}}. 
    In Figure {\ref{fig:rayleigh}}, we present the same results for a Rayleigh fading scenario, where each entry of the channel gain matrix is drawn i.i.d. from a Rayleigh distribution with unit variance.
}  
We compare our algorithm to zero-forcing with perfect CSI as well as to the gradient-descent based approach given as Algorithm 1 in \cite{dean2018blind}, which we henceforth refer to as simply `gradient descent'.  We also compare our algorithm to maximum-likelihood decoding with imperfect CSI.  
\change{
    As described in Section {\ref{sub:noise}}, our algorithm may fail when attempting to recover the channel gain in the presence of AWGN.  In this set of simulations, an outage is determined by the failure rate of our algorithm. 
    We observe that both the failure rate and BER performance of our algorithm is highly related to the condition number of the channel.  Removing the cases where our algorithm fails is nearly equivalent to removing the channels that are the most poorly conditioned.
    Additionally, we note that for $n=4$, the average Rayleigh channel has a higher condition number than the average Gaussian channel, which explains why the performance of our algorithm in Figure {\ref{fig:rayleigh}} is superior to the performance in Figure {\ref{fig:awgn}}
}
For all SNR values tested, we outperform ML decoding with an estimation error modeled as AWGN with variance as low 1\% of the channel noise variance.
We see that at low SNR, our algorithm nearly matches zero-forcing. At high SNR, our approach yields an increased BER over schemes like zero-forcing.  This is caused by the rounding procedure used in our vertex finding algorithm presented in Section {\ref{sub:noise}}.  We discuss why performance degradation arises in more detail in Section {\ref{sub:ber}}.  It is an open problem to develop alternative approaches for high SNR that mitigate this noise enhancement.     
\change{
    Obtaining analytic results regarding the performance of our algorithm in the presence of noise appears extremely difficult as one must account for randomness in the algorithm, the channel gain matrix, the transmitted symbols, and the AWGN.
}

\begin{figure*}
\hfill
\begin{tikzpicture}
\begin{semilogyaxis}[
    grid=both,
    width=7cm, height=7cm,
    xmax = 27,
    ymin = 1e-5,
    ymax = 1,
    xlabel={SNR (dB)},
    ylabel={BER},
    title={2.5\% Outage Probability},
    legend style={font=\scriptsize,
    	cells={anchor=west},
    	anchor=south west,
        at={(0,0)}}
]

\addplot[red,dashed,thick,mark=triangle*,
mark options={solid,mark size=3}] coordinates
{
(30, 0.012)
(23, 0.030)
(20, 0.046)
(17, 0.0629)
(13, 0.081)
(10, 0.1097)
(7, 0.1726)
};

\addplot[blue,thick,mark=*] coordinates
{
(27,2.1e-4)
(23,2.5e-4)
(20,4.9e-4)
(18,1.75e-3)
(17,0.0028)
(13,0.0334)
(10,0.0969)
(7,0.2175)
};

\addplot[black,dashed,thick] coordinates
{
(23,1e-6)
(20,4.0690e-6)
(18,2.95e-4)
(17,0.0015)
(13,0.0334)
(10,0.0969)
(7,0.2175)
};

\addplot[black,mark=square*,mark options=
{fill=black,mark size=2}] coordinates
{
(23,2e-6)
(20,1.8717e-4)
(18,0.0011)
(17,0.0031)
(13,0.0295)
(10,0.0959)
(7,0.2272)
};

\legend{ML, 1\% CSI Err.\\ Vertex Hopping\\ Grad. Descent [1]\\ ZF, Perfect CSI\\}
\end{semilogyaxis}
\end{tikzpicture}
\hspace{1.6cm}
\begin{tikzpicture}
\begin{semilogyaxis}[
    grid=both,
    width=7cm, height=7cm,
    xmax = 30,
    ymin = 1e-5,
    ymax = 1,
    xlabel={SNR (dB)},
    ylabel={BER},
    title={Zero Outage Probability},
    legend style={font=\scriptsize,
    	cells={anchor=west},
    	anchor=south west,
        at={(0,0)}}
]

\addplot[red,dashed,thick,mark=triangle*,
mark options={solid,mark size=3}] coordinates
{
(30, 0.040)
(24, 0.046)
(20, 0.056)
(17, 0.0629)
(14, 0.0765)
(12, 0.0842)
(11, 0.0924)
(10, 0.1117)
(9, 0.124)
(7, 0.168)
};

\addplot[blue,thick,mark=*] coordinates
{
(30,2.59e-3)
(23,9.2e-3)
(20,1.86e-2)
(17,3.16e-2)
(13,7.4e-2)
(10,0.111)
(7,0.136)
};

\addplot[black,dashed,thick] coordinates
{
	
(30, 0.00196126302083333)
(23, 0.007478841145833331)
(20, 0.0146375868055555)
(17, 0.0277669270833333)
(13, 0.0772840711805556)
(10, 0.143223741319444)
(7, 0.3)
};

\addplot[black,mark=square*,mark options=
{fill=black,mark size=2}] coordinates
{
(30, 0.000817)
(23, 0.005309)
(20, 0.0119574652777778)
(17, 0.0242214626736111)
(13,0.0693006727430556)
(10,0.147680664062500)
(7,0.2)

};

\legend{ML, 1\% CSI Err.\\ Vertex Hopping\\ Grad. Descent [1]\\ ZF, Perfect CSI\\}
\end{semilogyaxis}
\end{tikzpicture}
\hspace{5mm}
\caption{The BER using various MIMO decoding schemes to estimate $\mathbf{X}$ in the presence of AWGN where the channel gain matrix is Rayleigh distributed, with $n=4$ and $k=30$.}
\label{fig:rayleigh}
\end{figure*}
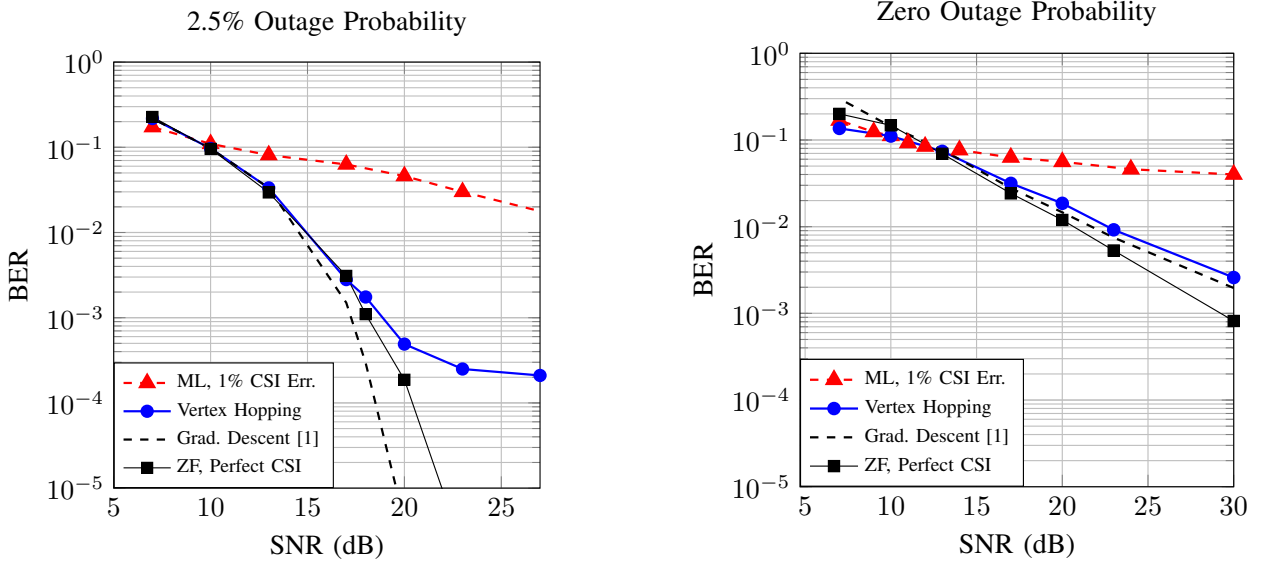

In Figure {\ref{fig:success}} we give the success probability of the vertex hopping algorithm for various values of $n$ in the limit of high SNR, where success is defined as properly recovering $\mathbf{X}$ correctly up to an ATM.  Here $\mathbf{X}$ is generated uniformly at random, and the success probability of our algorithm almost exactly matches the empirical probability that a random $\mathbf{X}$ has the correct theoretical guarantees provided in \cite{dean2018blind}.  We outperform the gradient descent algorithm in \cite{dean2018blind} beyond $n=5$.

\begin{figure}
\centering
\begin{tikzpicture}
\begin{axis}[
    grid=both,
    width=7cm, height=7cm,
    ymin = 0,
    ymax = 1,
    xlabel={Samples ($k$)},
    ylabel={Success Probability},
    legend style={font=\scriptsize,
    	cells={anchor=west},
    	anchor=south east,
        at={(0.95,0.05)}}
]
\addplot coordinates
{
(1, 0)
(2.000000, 0.512000)
(4.000000, 0.878000)
(6.000000, 0.966000)
(8.000000, 0.991000)
(10.000000, 0.999000)
(12.000000, 0.999000)
(14.000000, 1.000000)
(16.000000, 1.000000)
(18.000000, 1.000000)
(20.000000, 1.000000)
(22.000000, 1.000000)
(24.000000, 1.000000)
(26.000000, 1.000000)
(28.000000, 1.000000)
(30.000000, 1.000000)
(32.000000, 1.000000)
(34.000000, 0.999000)
};

\addplot coordinates
{
(2, 0)
(3.000000, 0.149000)
(5.000000, 0.478000)
(7.000000, 0.698000)
(9.000000, 0.846000)
(11.000000, 0.908000)
(13.000000, 0.944000)
(15.000000, 0.969000)
(17.000000, 0.981000)
(19.000000, 0.991000)
(21.000000, 0.996000)
(23.000000, 0.992000)
(25.000000, 0.995000)
(27.000000, 0.998000)
(29.000000, 0.998000)
(31.000000, 0.999000)
(33.000000, 0.999000)
(35.000000, 1.000000)
};

\addplot coordinates
{
(3,0)
(4.000000, 0.011000)
(6.000000, 0.199000)
(8.000000, 0.446000)
(10.000000, 0.648000)
(12.000000, 0.772000)
(14.000000, 0.849000)
(16.000000, 0.911000)
(18.000000, 0.959000)
(20.000000, 0.978000)
(22.000000, 0.983000)
(24.000000, 0.983000)
(26.000000, 0.990000)
(28.000000, 0.995000)
(30.000000, 0.998000)
(32.000000, 0.998000)
(34.000000, 0.996000)
};

\addplot coordinates
{
(5.000000, 0.001000)
(7.000000, 0.053000)
(9.000000, 0.270000)
(11.000000, 0.465000)
(13.000000, 0.637000)
(15.000000, 0.794000)
(17.000000, 0.890000)
(19.000000, 0.930000)
(21.000000, 0.961000)
(23.000000, 0.972000)
(25.000000, 0.991000)
(27.000000, 0.987000)
(29.000000, 0.997000)
(31.000000, 0.995000)
(33.000000, 0.995000)
(35.000000, 0.998000)
};

\addplot coordinates
{
(6.000000, 0.000000)
(8.000000, 0.017000)
(10.000000, 0.129000)
(12.000000, 0.341000)
(14.000000, 0.530000)
(16.000000, 0.705000)
(18.000000, 0.824000)
(20.000000, 0.876000)
(22.000000, 0.940000)
(24.000000, 0.956000)
(26.000000, 0.980000)
(28.000000, 0.984000)
(30.000000, 0.988000)
(32.000000, 0.996000)
(34.000000, 0.993000)
};

\addplot coordinates
{
(7.000000, 0.000000)
(9.000000, 0.008000)
(11.000000, 0.064000)
(13.000000, 0.210000)
(15.000000, 0.387000)
(17.000000, 0.552000)
(19.000000, 0.742000)
(21.000000, 0.815000)
(23.000000, 0.906000)
(25.000000, 0.925000)
(27.000000, 0.956000)
(29.000000, 0.971000)
(31.000000, 0.976000)
(33.000000, 0.989000)
(35.000000, 0.990000)
};

\addplot coordinates
{
(8.000000, 0.000000)
(10.000000, 0.000000)
(12.000000, 0.022000)
(14.000000, 0.085000)
(16.000000, 0.275000)
(18.000000, 0.475000)
(20.000000, 0.641000)
(22.000000, 0.788000)
(24.000000, 0.857000)
(26.000000, 0.914000)
(28.000000, 0.939000)
(30.000000, 0.967000)
(32.000000, 0.977000)
(34.000000, 0.981000)
};

\legend{n=2,n=3,n=4,n=5,n=6,n=7,n=8}
\end{axis}
\end{tikzpicture}
\caption{The success rate of our approach when $\mathbf{X}$ is drawn uniformly at random for various values of $n$. Our vertex hopping algorithm succeeds with near 100\% likelihood when $\mathbf{X}$ has the theoretical properties to ensure recovery, given in \cite{dean2018blind}.}
\label{fig:success}
\end{figure}
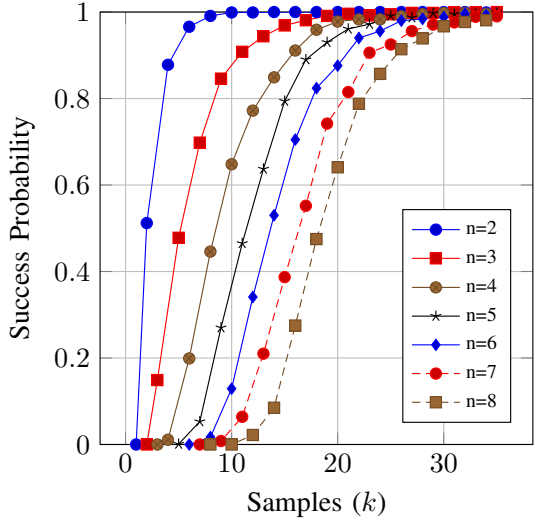

\subsection{Performance Analysis}
\label{sub:performance}
A rigorous theoretical analysis of the time complexity of our algorithm is beyond the scope of this paper.
However, we informally comment on the typical performance observed by our algorithm.
Our algorithm treats the blind decoding problem as a mixed-integer linear program; optimizing such programs is classic example of an NP-complete problem \cite{karp1972reducibility}.
Thus, it is likely that exactly solving ({\ref{eq:obj}}) -- ({\ref{eq:const}}) becomes hard for large $n$.


We note, however, that in case that occurs with non-negligible probability for $n \leq 4$, the initialization step given in Section {\ref{sub:initial}} will directly return a global optima.
The time complexity of this step of our algorithm is considered in more depth in Section \ref{sub:initial}, where it is shown to be $O\left(n^4 k\right)$.  
Further, we have observed that for $n \leq 8$, if the problem is not solved by the initialization step, we typically require only one or two vertex hops.  As described in Section \ref{sub:pivot}, a single iteration of the pivoting process requires only $O\left( n^3 \right)$ operations.
Thus in practice, our algorithm typically performs near the limit of $O\left(n^4 k\right)$.
Empirically, we find that for $n \leq 5$, the runtime of our algorithm indeed scales roughly linearly with $k$.  
For larger values of $n$, we discuss the dependence of the runtime on $k$ in Section {\ref{sub:trap}}.

For $n \geq 6$, where local optima exist, we occasionally encounter a `trap' case as described in Section \ref{sub:tableau}.  
In these cases the solver often enumerates a large subspace before exiting, thus inflating the average runtime.
Beginning at $n = 8$, the runtime of our algorithm begins to noticeably increase with $n$. 
We note that the spectrum of possible determinants of $\pm 1$-valued matrices grows rapidly as $n$ grows (see \cite{orrick_spectrum}) and we conjecture that this rapid growth of possible determinants leads to the increased runtime of our algorithm with $n$.

\section{Initialization}
\label{sub:initial}
Linear programs are often solved by the simplex technique which hops between vertices of the feasible region in order to find a global optimum.  Typically the simplex technique begins at the origin of the feasible region.
When the origin of a linear program is not feasible,
techniques exist to find a suitable feasible solution, often termed a basic feasible solution or BFS \cite{papadimitriou1998combinatorial}.
Unfortunately, it is not clear how to leverage standard techniques to solve our problem.
In our case, the origin is singular, as are a majority of the vertices of the feasible region, and so in particular the gradient is not defined and we cannot start our vertex-hopping technique at these points.

We now present our algorithm for finding a suitable BFS, which is summarized in Algorithm {\ref{alg:init}} and is referred to as the `vertex finding' process.
The $\ell_\infty$ constraints in \pref{eq:const} can be expressed as $2kn \times n^2$ linear inequality constraints: let $\mathbf{u} = \text{vec}(\mathbf{U})$, then \pref{eq:const} can be expressed as $\bar{\mathbf{Y}}\mathbf{u} \leq \mathbbm{1}$ for some appropriate $\bar{\mathbf{Y}}$.  Suppose $\mathbf{U}$ satisfies $l$ constraints with equality. Then we form the $l \times n^2$ matrix $\mathbf{B}$ by taking the appropriate rows of $\bar{\mathbf{Y}}$.  The matrix $\mathbf{B}$ now describes the active constraints.  If we consider $\tilde{\mathbf{U}} = \mathbf{U} + \Delta$ and $\Delta \in \text{null}(\mathbf{B})$, then $\tilde{\mathbf{U}}$ will still, at a minimum, satisfy the same $l$ constraints with equality. We refer to the process of setting  $\tilde{\mathbf{U}} = \mathbf{U} + \Delta$ for some $\Delta \in \text{null}(\mathbf{B})$ as ``moving within the nullspace of the active constraints''.

We begin by choosing a random feasible point using the technique described in \cite{dean2018blind}, \change{which involves drawing $\mathbf{U}$ uniformly from the set of orthogonal matrices of order $n$ and then scaling $\mathbf{U}$ until $\mathbf{UY}$ is feasible}. Our technique successively solves one-dimensional optimization problems. At each step, $\mathbf{U}$ is perturbed so that at least one additional constraint becomes active.  
This is accomplished by projecting the gradient onto the nullspace of the active constraints.  
The solver then moves in this direction,
leaving the already active constraints unchanged, until an additional constraint becomes active.
As a result, when Algorithm \ref{alg:init} terminates, we are guaranteed that at least $n^2$ constraints are active.  A graphic depiction of this process is shown in Figure \ref{fig:dynamic}. 

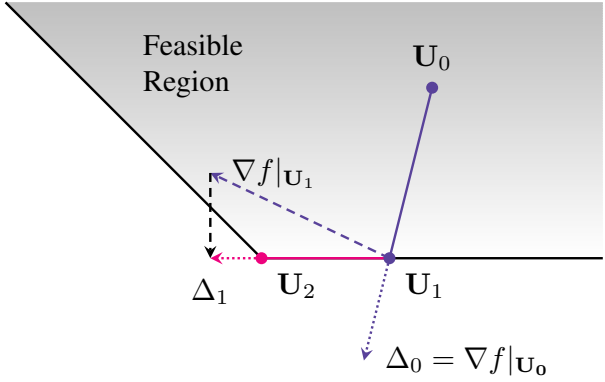
\begin{figure}
\centering
\begin{tikzpicture}
    \shade[top color = lightgray,
           bottom color = white]
           (0,0) -- (4,0) -- 
           (4,3) -- (-3,3) -- cycle;

    \draw[-,thick] (0,0) -- (4,0);
    \draw[-,thick] (0,0) -- (-3,3);
    
    \node[align=left] at (-0.8,2.25) {Feasible \\ Region};

    \node[circle,Violet, fill=Violet,inner sep=0pt,minimum size=4pt,label=above:
		{$\mathbf{U}_0$}] (a) at (2,2) {};

	\draw[-,Violet, thick] (2,2) -- (1.5,0);
	\draw[-stealth,densely dotted, thick,Violet] (1.5,0)--(1.2,-1.2);
	\draw[-stealth,thick, densely dashed,Violet] (1.5,0) -- (-0.6,1);

	\draw[-, thick, RubineRed] (1.5,0) -- (0,0);
	\draw[-stealth,densely dotted, thick,RubineRed] (0,0) -- (-0.6,0);
	\draw[-stealth, thick, densely dashed, black] (-0.6, 1) -- (-0.6,0);

	\node[label=right:{$\nabla f \rvert_{\mathbf{U}_1}$}] at (-0.6,1) {};
	\node[label=right:
	{$\Delta_0 = \nabla f \rvert_{\mathbf{U_0}}$}] 
	at (1.2,-1.2) {};
	\node[label=below:{$\Delta_1$}] at (-0.6,0) {};

	\node[circle,Violet, fill=Violet,inner sep=0pt,minimum size=4pt,label=below right:
		{$\mathbf{U}_1$}] (b) at (1.5,0) {};

	\node[circle,RubineRed, fill=RubineRed,inner sep=0pt,minimum size=4pt,label=below right:
		{$\mathbf{U}_2$}] (c) at (0,0) {};
		
\end{tikzpicture}
\caption{A graphical depiction of the vertex finding process. $\mathbf{U}_0$ is chosen uniformly at random.  We then proceed in the
direction given by the gradient until we reach the boundary of the feasible
region.  The next step is taken by projecting the gradient at $\mathbf{U}_1$
onto the nullspace of the basis formed by the set of active constraints.}
\label{fig:dynamic}
\end{figure}

\begin{algorithm}
 \caption{Vertex Finding}
 \label{alg:init}
 \begin{algorithmic}[1]
 \renewcommand{\algorithmicrequire}{\textbf{Input:}}
 \renewcommand{\algorithmicensure}{\textbf{Output:}}
 \REQUIRE An $n \times k$ matrix of received samples $\mathbf{Y}$. Any feasible, full rank starting point $\mathbf{U}$.
 \ENSURE  A matrix $\mathbf{U}$ that satisfies at least $n^2$ linearly independent constraints of \pref{eq:const} with equality.
  \WHILE {$\mathbf{B}$ is not full rank}
  \STATE $\Delta =\mathbf{U}^{-\intercal}$.
  \STATE Find $\mathbf{N}$, a basis for $\text{null}(\mathbf{B})$.
  \STATE Compute $\Delta = \text{proj}_\mathbf{N} \Delta$.
  \STATE Find max $t \in \mathbb{R}_+$ such that \\ $\| (\mathbf{U} + t \Delta) \mathbf{Y} \|_\infty = \pm1$.
  \STATE Update $\mathbf{B}$.
 \ENDWHILE
 \RETURN $\mathbf{U}$
 \end{algorithmic} 
 \end{algorithm}

 \begin{algorithm}
 \caption{Efficient Projection}
 \label{alg:basis}
 \begin{algorithmic}[1]
 \renewcommand{\algorithmicrequire}{\textbf{Input:}}
 \renewcommand{\algorithmicensure}{\textbf{Output:}}
 \REQUIRE The $l \times n^2$ matrix $\mathbf{B}$, the $n \times n$ matrix $\Delta$.
 \ENSURE  The component of $\Delta$ that lies in $\text{null}(\mathbf{B})$.
  \STATE Compute $\mathbf{C=B B}^\intercal$.  Save indices of non-zero off-diagonal entries of $\mathbf{C}$.
  \FOR {each $i<j$ such that $c_{ij} \neq 0$}
  \STATE $\mathbf{u} = \mathbf{b}^{(i)}$, 
  		 $\mathbf{v} = \mathbf{b}^{(j)}$
  \STATE $\mathbf{v} = \mathbf{v} - \left(\mathbf{u} \mathbf{v}^\intercal / \mathbf{u u}^\intercal \right) \mathbf{u}$
  \IF {$\| \mathbf{v} \| > 0$}
  	\STATE $\mathbf{b}^{(j)} = \mathbf{v} / \| \mathbf{v} \|$
  \ELSE
  	\STATE Store $j$ as redundant.
  \ENDIF
 \ENDFOR
 \STATE Remove redundant rows of $\mathbf{B}$.
 \STATE $\mathbf{\delta} = \text{vec}(\mathbf{\Delta})$
 \FOR {each row $\mathbf{b}^{(i)}$ in $\mathbf{B}$}
 	\STATE $\mathbf{\delta} = \mathbf{\delta} - \left( \mathbf{\delta} \cdot
\left(\mathbf{b}^{(i)}\right)^\intercal \right) \mathbf{b}^{(i)}$
 \ENDFOR
 \RETURN $\text{mat}(\delta)$
 \end{algorithmic} 
 \end{algorithm}

Algorithm \ref{alg:init} will return a value of $\mathbf{U}$ such that at least $n^2$ entries of $\mathbf{UY}$ will be equal to $\pm 1$, with at least $n$ of these entries per row.   
Because the remaining entries of $\hat{\mathbf{X}}=\mathbf{UY}$ are not
independent, we often find that most, if not all, entries equal $\pm 1$.  
In Section \ref{sub:trap} we discuss how often Algorithm \ref{alg:init} returns a
suitable $\hat{\mathbf{X}}$, and more complete empirical results regarding the
 distribution of the elements of $\mathbf{UY}$ at the output of Algorithm \ref{alg:init} are given in Section \ref{sub:noise}.

After running Algorithm \ref{alg:init}, the resulting columns of $\hat{\mathbf{X}}$ must be separated into `good' versus `bad' columns.  Good columns are those that lie in $\{-1,+1\}^n$.  We let $\mathbf{Y}_g$ and $\mathbf{Y}_b$ denote the matrices composed of corresponding good and bad columns of $\mathbf{Y}$ respectively.
If $\mathbf{Y}_g$ is not full rank, then we must rerun Algorithm \ref{alg:init} again until a suitable $\hat{\mathbf{X}}$ is obtained.  For sufficiently large $k$, it is very rare in the noiseless case that Algorithm \ref{alg:init} fails to produce a suitable output; we further quantify when this happens in Section \ref{sub:trap}. 

We now briefly comment on the time complexity of Algorithm \ref{alg:init}.
Algorithm \ref{alg:init} will require $n^2$ iterations to ensure that at least $n^2$ constraints become active.  Computing the inverse of $\mathbf{U}$ requires $O\left(n^3\right)$ operations.  The process of projecting the gradient onto the nullspace of the active constraints is considered in Section \ref{sub:projection} where it is shown that this requires at most $O\left( n^3 \right)$ operations per call. Finally, updating $\mathbf{B}$ requires computing the product $\mathbf{\Delta Y}$ which requires $O\left( n^2 k \right)$ operations.  Since $k > n$, this means that Algorithm \ref{alg:init} requires $O\left(n^4 k\right)$ operations.

\subsection{Efficient Projection}
\label{sub:projection}
At each iteration of Algorithm \ref{alg:init}, we must project the gradient onto the
nullspace of the active constraints.  Since the matrix $\mathbf{B}$ has $n^2$ rows, a na\"ive approach to finding the nullspace of $\mathbf{B}$ would require $O(n^5)$ operations.  In this subsection, we describe how to efficiently exploit the sparse structure of the matrix $\mathbf{B}$ to efficiently project a vector away from $\mathbf{B}$.

We begin by considering the structure of the matrix $\mathbf{B}$. At each iteration of the inner loop of Algorithm \ref{alg:init}, at least one more constraint becomes active.  Suppose that there are $l$ active constraints at the beginning of an iteration of Algorithm \ref{alg:init}, and that at the completion of this iteration, $\hat{x}_{ij}$ (where $\hat{\mathbf{X}} = \mathbf{UY}$) becomes $\pm 1$. 
In this case, we append the following row to
$\mathbf{B}$:
\begin{equation*}
\mathbf{b}^{(l)} = \left[ \mathbf{0}_{i*n} \quad  {y^{(j)}}^{\intercal} \quad
\mathbf{0}_{(n-i-1)*n} \right].
\end{equation*}
That is, the $l$th row of $\mathbf{B}$ will consist of $i*n$ zeros, followed by
the $j$th column of $\mathbf{y}$, with zeros in the remaining entries.  If multiple constraints become active in a single iteration, we simply add additional rows to $\mathbf{B}$ in the same manner.  
From this discussion, it is clear that the rows of $\mathbf{B}$ can be ordered to obtain a block diagonal structure.  At most, $\mathbf{B}$ will consist of $n$ blocks that are $l \times n$ large.  In each iteration of Algorithm \ref{alg:init}, we can easily insert each row of $\mathbf{B}$ appropriately to maintain this block structure.

Given this block matrix, an efficient procedure to find the component of $\Delta$ that lies in $\text{null}(\mathbf{B})$ is described in Algorithm \ref{alg:basis}.  The algorithm first finds an orthonomoral basis for the subspace $\text{span}(\mathbf{B})$.  The block structure of $\mathbf{B}$ means that many of its rows are already orthogonal, this means the matrix $\mathbf{C}=\mathbf{BB}^\intercal$ is nearly diagonal.  Given a basis for $\text{span}(\mathbf{B})$, we obtain the desired output by projecting $\text{vec}\left(\Delta\right)$ away from each basis vector.

We now consider the runtime of this procedure.
The matrix $\mathbf{C}$ need not be updated from scratch at each step and can be computed block-wise.  We also note that computing $\mathbf{C}$ gives all inner products needed in line 4 of Algorithm \ref{alg:basis}, and each inner product requires $O(n)$ operations. At most we will need to compute $n^2$ of these inner products per call to Algorithm \ref{alg:basis}.  Similarly, for the loop in lines 13-15, the matrix $\mathbf{B}$ will have at most $n^2$ rows and each vector rejection operation will require $O(n)$ operations.  At worst, this process will require $O(n^3)$ operations.  However, in practice, we find that for most calls to Algorithm \ref{alg:basis}, each block in $\mathbf{C}$ has fewer than $n$ rows and $\mathbf{C}$ will only have one or two non-zero entries.  Thus we typically perform much better than this bound.

\section{Vertex Hopping}
\label{sec:vertex}
\subsection{Graph of Vertices}
\label{sub:graph}
We know that solutions of the blind decoding problem will only occur on vertices and thus we attempt solve the blind decoding problem by searching these vertices.  We note that there are at most $2^{n^2}$ possible vertices $\mathbf{UY}$ where $\mathbf{U} \in \mathbb{R} ^ {n \times n}$ and $\mathbf{Y} \in \mathbb{R} ^ {n \times k}$.  
Note that if $k > n$ strictly, then a vertex is determined entirely by any $n$ linearly independent columns of $\mathbf{Y}$. Additionally, if $k > n$, then some of the $2^{n^2}$ vertices may become infeasible.

Our search process, summarized in Algorithm {\ref{alg:main}}, begins with a vertex determined by Algorithm {\ref{alg:init}}.  Given such a point we can restrict ourselves to performing optimization on a graph that represents the vertices of the feasible region.
In this graph, there are at most $2^{n^2}$ nodes corresponding to vertices of the feasible region as described above.  
An edge exists between two nodes if their corresponding values of $\mathbf{UY}$ differ by Hamming distance $1$.
Thus, our graph contains up to $2^{n^2}$ nodes and $n^2 \cdot 2^{n^2}$ edges.

Specifically, given a value of $\mathbf{U}$ obtained by the output of Algorithm \ref{alg:init}, we attempt to solve the blind decoding problem by solving the following program
\begin{align}
\text{maximize} \quad & \log |\det \mathbf{U}| \label{eq:obj2}\\
\text{s.t.} \quad & \mathbf{UY}_g = \pm 1 \label{eq:const2} \\
				  & \|\mathbf{UY}_b\|_\infty \leq 1. \nonumber 
\end{align}
This is a non-linear mixed-integer program and we will attempt to optimize it as such.
We do so by flipping the signs of individual entries of $\mathbf{UY}_g$ in an attempt to hillclimb towards an optimal value of $\mathbf{U}$ while allowing for backtracking when we reach a local optimum. Before we discuss our approach to solving it, we give state several important facts about the program.
\begin{claim}
All global optima of \pref{eq:obj2}--\pref{eq:const2} are also global optima of \refprog. 
\end{claim}
\begin{proof}
We know that \refprog contains global maxima such that $\mathbf{UY} \in \{-1, +1\}^{n \times k}$ \cite[Lemma 3]{dean2018blind}.
Any $\mathbf{U}$ such that $\mathbf{UY} \in \{-1, +1\}^{n \times k}$ will clearly be feasible in both \refprog and \pref{eq:obj2}--\pref{eq:const2}. 
So there must be a $\mathbf{U}$ that maximizes \refprog that also maximizes \pref{eq:obj2}--\pref{eq:const2}.  
Further, any $\mathbf{U}$ that is feasible in \pref{eq:obj2}--\pref{eq:const2} is clearly also feasible in \refprog. 
So all $\mathbf{U}$ that maximize \pref{eq:obj2}--\pref{eq:const2} will also maximize \refprog.
\end{proof}

We note that the converse of this statement is not true: when $n$ is such that no Hadamard matrices exist, \refprog may contain (non-strict) optima that are not feasible in \pref{eq:obj2}--\pref{eq:const2}, see \cite{dean2018blind}.
We also note that we have not ruled out the possibility that the additional
constraints imposed in \pref{eq:obj2}--\pref{eq:const2} introduce local optima that
are not present in \refprog.
Empirically, we have not observed such optima for $n < 6$.
Further, for $n \geq 6$, the fraction of local optima that are encountered while optimizing
\pref{eq:obj2}--\pref{eq:const2} is consistent with the fraction of optima that
are local in \refprog. 
This suggests that either such spurious optima  do not exist or are insignificant in number.
However, we defer on obtaining analytic results
supporting this claim.

\blockcomment{
		\begin{claim}
		\label{claim:local}
		The program given by \pref{eq:obj2}--\pref{eq:const2} may contain local optima that are not optima of \refprog.
		\end{claim}
		\begin{proof}
		This is best illustrated with a counterexample.  Let $n=5, k=9$. Let $\mathbf{A}=\mathbf{I}$ and
		\begin{equation*}
		\mathbf{X} = 
		\begin{bmatrix}
		 -1 & -1 & -1 & -1 & -1 & -1 & -1 & -1 & -1 \\
		 -1 & -1 & -1 &  1 &  1 & -1 &  1 &  1 & -1 \\
		 -1 & -1 &  1 & -1 &  1 &  1 &  1 & -1 &  1 \\
		 -1 &  1 & -1 & -1 &  1 &  1 &  1 &  1 &  1 \\
		 -1 &  1 &  1 &  1 & -1 & -1 &  1 &  1 &  1
		\end{bmatrix}.
		\end{equation*}
		Suppose that Algorithm \ref{alg:init} returns
		\begin{equation*}
		\mathbf{U} =
		\begin{bmatrix}
		 0 &  0 &  0 & -1 &  0 \\
		 0 &  0 &  0 &  0 &  1 \\
		-2 & -1 & -1 &  0 & -1 \\
		 0 & -1 &  0 &  0 &  0 \\
		 0 &  0 &  1 &  0 &  0
		\end{bmatrix},
		\end{equation*}
		then we have 
		\begin{align*}
		\mathbf{Y}_g &= 
		\begin{bmatrix}
		  1 & 1& -1& -1 &  1 &-1  \\
		  1 &-1&  1&  1&  -1& -1 \\
		 -1 & 1&  1&  1&  -1&  1 \\
		  1 & 1&  1&  1&   1&  1 \\
		 -1 &-1& -1&  1&   -1& 1 \\
		\end{bmatrix}, \\
		\mathbf{Y}_b &= 
		\begin{bmatrix}
		0  & 0 & 0 \\ 
		-1 & -1 & -1 \\ 
		-1 & 1 & 1\\ 
		1 & 1 & 1\\
		1 & 1 & -1
		\end{bmatrix}.
		\end{align*}
		This is a local optima of \pref{eq:obj2}--\pref{eq:const2}.  One can see this changing each entry of $\mathbf{Y}_g$ individually. 
		For each entry, either $\mathbf{Y}_b$ becomes infeasible, $\mathbf{Y}_g$ (and hence the corresponding value of $\det \mathbf{U}$) becomes singular, or the value of $\det \mathbf{U}$ remains unchanged.  
		Hence, this is a local optima of \pref{eq:obj2}--\pref{eq:const2}.  This cannot be a local optima of \refprog because \refprog has no local optima for $n=5$ \cite[Lemma 10]{perlstein2018n5case}.
		\end{proof}
}


\subsection{Tableau Formation}
\label{sub:tableau} 
We use the tableau data structure, commonly used to implement the simplex algorithm, to efficiently allow us to `hop' between feasible values of $\mathbf{U}$ and flip a single entry of $\mathbf{UY}_g$.  
In this subsection, we describe how to formulate this tableau in a process that closely follows \cite{papadimitriou1998combinatorial}.
  
Given an output of Algorithm \ref{alg:init}, $\mathbf{U}$, we construct a linear program (LP) that is the first-order approximation of \pref{eq:obj2}--\pref{eq:const2}.  To do this, we simply replace the objective function with the gradient of the objective function, which in this case is $\mathbf{U}^{-\intercal}$.  We do not fully solve this LP as the objective function must be updated after each simplex hop. In order to form a traditional simplex tableau, the problem must be expressed in \emph{standard form}.  An LP in standard form optimizes a linear functional over non-negative vector $\mathbf{x}$ subject to a series of equality constraints, explicitly, for some $\mathbf{b}, \mathbf{c} \in \mathbb{R}^n$ and $\mathbf{A} \in \mathbb{R}^{n \times n}$, we aim to find $\mathbf{x} \in \mathbb{R}^n$ according to
\begin{align*}
\underset{\mathbf{x}}{\text{max}} \quad  & \mathbf{c}^\intercal \mathbf{x} \\
\text{s.t.} \quad & \mathbf{Ax} = \mathbf{b} \\
          		  & \mathbf{x} \geq 0.
\end{align*}
In order to express the constraints in \pref{eq:const2} in standard form, we must replace each entry $u_{i,j}$ with a pair of variables constrained such that $u_{i,j} = x_{i'} -x_{i''}$ and $x_i \in \mathbb{R}_+$ for all $i$.  The $\ell_\infty$ constraints in {\pref{eq:const2}} can be replaced by a series of inequality constraints as described in Section {\ref{sub:initial}}.  Each of these inequality constraints can further be replaced by a single equality constraint with the addition of a slack variable;
that is, a constraint in the form $\sum_i a^{(i)} x_i \leq b_i$ becomes $\sum_i a_i x_i + x_j = b_i$ for some $x_j \geq 0$.  If $\mathbf{Y}_g$ has $l$ columns, then we are left with a set of $2nk$ linear equations with $2n(n+l)$ variables.

The final step in forming a tableau is to solve for all $x$ variables to obtain a suitable BFS.
First, for each $x_{i'}$ and $x_{i''}$ such that $u_{i,j} = x_{i'} -x_{i''}$, we must set exactly one $x$ variable to zero such that the remaining $x$ variable is positive (that is, if $u_{i,j}$ is positive, then $x_{i'} = u_{i,j}$ and $x_{i''}=0$).
We then perform Gauss-Jordan elimination on the remaining linear equations in order to solve for the remaining $x$ variables. 
We note that the sparse structure of the tableau can be exploited to efficiently perform this elimination, see \cite{papadimitriou1998combinatorial} for a more complete description of this process.

\subsection{Searching the Feasible Region}
\label{sub:pivot}
Hopping to an adjacent vertex of the feasible region involves flipping exactly one entry of $\mathbf{UY}_g$ between its two possible extrema, namely $\{-1, +1\}$.  
The simplex tableau allows us to perform the necessary calculations efficiently; changing an entry of $\mathbf{UY}$ affects only a single row of $\mathbf{U}$, implying that each pivot involves manipulating only a small subset of the linear equations and variables that make up the full tableau. 
Indeed, computing the corresponding change in $\mathbf{U}$ for each neighbor can be done in only $O(n)$ operations.
We defer to existing literature on the simplex algorithm (for example \cite{papadimitriou1998combinatorial}) for a description of the pivoting process.  
Instead, in this subsection we focus on the process of efficiently selecting an appropriate vertex to pivot to at each step of the depth-first search.

Unlike ordinary simplex, the process of selecting a pivot is slightly more involved in our algorithm.
Each vertex has at most $n^2$ feasible neighbors.  
The structure of the tableau allows us to easily determine the change that hopping to each of these vertices will induce on $\mathbf{U}$.  Suppose we update the $i$th row of $\mathbf{U}$ such that $\mathbf{u}^{(i)} = \mathbf{u}^{(i)} + \Delta$ for some $\Delta \in \mathbb{R}^{1 \times n}$, meaning $\mathbf{U} = \mathbf{U} + \mathbf{e}_i \Delta$.  We can use the matrix determinant lemma (\cite{ding2007eigenvalues}) to compute the corresponding change in the objective function as
\begin{align}
\log | \det \left( \mathbf{U} + \mathbf{e}_i \Delta \right) | &=
\log | \left(1+\Delta \mathbf{U}^{-1} \mathbf{e}_i \right)
\det \left( \mathbf{U} \right) | \nonumber \\
&= \log | \left(1+\Delta \mathbf{U}^{-1} \mathbf{e}_i \right) | + \log | \det \left( \mathbf{U} \right) | \nonumber \\
\label{eq:neighbor}
\end{align}

Because the logarithm function is monotonic, we can easily predict the change this hop will impose on the objective function by simply considering the value of $\left| 1 + \Delta \mathbf{U}^{-1} \mathbf{e}_i \right|$.  This requires only $O(n)$ operations.
The process of inspecting all $n^2$ neighbors and subsequently pivoting to an optimal choice can be accomplished in only $O(n^3)$ operations.
Moreover, because each pivot imposes a rank-one change in $\mathbf{U}$,
$\mathbf{U}^{-1}$ can be efficiently updated via the Sherman-Morrison inverse
formula \cite{sherman1950adjustment} after each hop, meaning $\mathbf{U}^{-1}$
need not be computed from scratch each step.

During the vertex hopping process, we select the vertex in the direction of the maximal positive gradient that has not already been visited, continuing in a depth-first-search manner until a global optimum has been located or a preset hop limit has been exhausted. 
For many values of $n$, there may exist local optima;
for this reason, we allow the solver to hop to neighboring vertices where the 
objective function is equal or lesser in value.
For most values of $n$, we can detect when we have reached a global optimum, allowing the solver to terminate without enumerating all non-singular vertices of the feasible region.  Our early termination criteria are discussed in Section \ref{sub:stop}. 

In order to efficiently backtrack when we reach a local optimum,
we keep a complete snapshot of the tableau at each of the previously visited vertices.  
While we do allow the solver to visit previously unvisited vertices that
decrease the objective function, we do not allow the solver to traverse vertices
where $\det \mathbf{U} = 0$.
This is because the gradient of the objective function is undefined at these
vertices.
While this may restrict the search space of the solver, this is not frequently 
an issue.
We discuss cases where this becomes an issue in Section \ref{sub:trap}.

Finally, we note that considering other criteria to select the next vertex in
our search is a topic of future research.
For example, it may be possible to further optimize the vertex hopping process by simply 
hopping to the first neighbor encountered that increases the objective function rather than 
exhaustively checking all $n^2$ neighbors. 
Additionally, when the fraction of `bad' columns is large, one may consider
positively weighting vertices which cause more columns to be `good'.  
For $n \leq 5$, our current vertex hopping strategy rarely requires more than
one or two hops to reach a global optima. 
However, as $n$ grows, for certain values of $k$ the majority of the time 
to obtain the solution is consumed by the vertex hopping process. 
Thus, improved vertex selection criteria may be necessary to further reduce the
runtime of our algorithm for $n > 5$.

\begin{algorithm}
 \caption{Simplex with Backtracking}
 \label{alg:main}
 \begin{algorithmic}[1]
 \renewcommand{\algorithmicrequire}{\textbf{Input:}}
 \renewcommand{\algorithmicensure}{\textbf{Output:}}
 \REQUIRE $\mathbf{U}$ such that $\mathbf{UY}_g$ is full rank.
 \ENSURE  $\hat{\mathbf{X}}$ or error.
  \STATE form simplex tableau
  \REPEAT
  \STATE compute objective function at all neighboring vertices \change{using the matrix determinant lemma}
  \IF {optimum reached}
  \IF {global optimum}
  \RETURN $\hat{\mathbf{X}}=\mathbf{UY}$
  \ELSE
    \STATE backtrack to first vertex with unvisited neighbor
  \ENDIF
  \ENDIF
  \STATE pivot to the largest feasible unvisited neighbor
  \STATE update $\mathbf{U}^{-1}$ \change{using the Sherman-Morrison inverse formula}
 \UNTIL {state space exhausted}
 \RETURN error
 \end{algorithmic} 
 \end{algorithm}
 
 \subsection{Stopping Criteria}
 \label{sub:stop}
Since we do not know $\mathbf{A}$, nor the value of $|\det \mathbf{A}|$, the value of $\det \mathbf{U}$ alone does not tell us whether \refprog has obtained a global optimum.  
However, the determinant of any square matrix with values $\pm 1$ can only take on a discrete set of values (see \cite{orrick_spectrum} or\cite{macwilliams2006theory}).  As a result, the only values that $|\det \mathbf{U}|$ can take is simply this spectrum scaled by some unknown constant, namely $\det \mathbf{A}^{-1}$. 
In other words, \refprog is maximized when the following quantity obtains the maximum determinant for any $\pm 1$-valued matrix of dimension $n$:
\begin{align}
\max_{\tilde{\mathbf{Y}}\in \mathbb{R}^{n \times n}} \quad &|\det \mathbf{U}\tilde{\mathbf{Y}}|
\label{eq:obj3}\\
\text{s.t.} \quad &\text{cols}(\tilde{\mathbf{Y}}) \subseteq \text{cols}(\mathbf{Y}). \label{eq:const3}
\end{align}

Solving this problem seems like a difficult combinatorial optimization task so we do not rely on solving it directly.
However, when we are at an optimum, we can often determine the value of \pref{eq:obj3} by inspecting the value of the objective function at neighboring vertices.
Indeed, for certain values of $n$, we may uniquely determine when we have reached a global optimum by considering the relative change between the value of the objective function at the optimum and at its $n^2$ neighbors.

In this subsection, we provide necessary and sufficient conditions to determine when Algorithm \ref{alg:main} has reached a global optima, for $n \in \{2,\ldots, 6, 8, 10, 12\}$.  
In Appendix \ref{apx:stop}, we prove that these conditions are sufficient up to $n=6$; the conditions given for $n=8, 10,$ and 12 are conjectured to be sufficient based on empirical results.

As stated in Section \ref{sub:graph}, each vertex may have up to $n^2$ neighbors, although not all these neighbors may be feasible.  Regardless of whether or not these vertices are feasible, we may still compute the objective function at these neighboring values using the process outlined in Section \ref{sub:pivot}.  The following criteria, describing the value of the objective function at all $n^2$ neighboring vertices, may be applied to uniquely determine when Algorithm \ref{alg:main} has reached a global optimum.

\begin{itemize}
\item For $n\leq 5$, there are no local optima.  
In these cases, if all $n^2$ neighbors of $\mathbf{UY}$ decrease the value of the objective function, then $\mathbf{UY}$ must be a global optimum.
\item For $n = 6$, at any global optimum, the value of the objective function at all 36 neighbors will take on exactly 3 distinct values.  
Further, only global optima have neighbors such that the value of the objective function decreases to a fraction of 4/5 of the optimal value.
\item For $n = 8$, if all 64 neighbors of $\mathbf{UY}$ decrease the value of the objective function to a fraction of 3/4 of the optimal value, then $\mathbf{UY}$ must be a global optimum.
\item For $n=10$, 20 neighbors of $UY$ must decrease the value of the objective function by a fraction of 1/3 the optimal value while the remaining 80 neighbors must decrease the value by 1/6.
\item For $n=12$, all 144 neighbors of $UY$ must decrease the objective function by 1/6 of the optimal value.
\end{itemize}
 
 \subsection{Causes of Failure}
\label{sub:trap}
In this subsection, we consider when Algorithms {\ref{alg:init}} (Vertex Finding) and
{\ref{alg:main}} (Simplex with Backtracking) fail to properly
terminate.  Provided that these algorithms are given proper inputs, the
probability that they succeed, for typical values of $n$ and $k$, is high. In
cases where we do encounter a failure, we find that rerunning the algorithm with
a new starting point ($\mathbf{U}_0)$ is often sufficient to recover from this
failure.  
For both Algorithms \ref{alg:init} and
\ref{alg:main}, we provide the average number of calls required to obtain a solution to the blind decoding problem for various values of $n$ and $k$ in Table \ref{tab:failure}, along
with the average runtime of each algorithm.  We outline the reasons both Algorithms \ref{alg:init} and \ref{alg:main} fail in the remainder
of this subsection.

Although not depicted in Table \ref{tab:failure}, we note that the number of calls to each subroutine is highly dependent the value of $k$.
This is evidenced in Figure \ref{fig:speedbyk}, 
which plots the average runtime of our solver, normalized by $k$,
for $n=6,8,10,$ and $12$.  As previously discussed, an optimistic estimate of the runtime of our
algorithm is $O(n^4 k)$.  
If our algorithm performed near this estimated
runtime, we would expect these lines would be constant, i.e. that the runtime
would scale linearly with $k$.  
For $n \geq 6$ and small values of $k$, our algorithm does not appear to perform near the best-case complexity. 
This is because, for large $n$ and small $k$, the output of Algorithm \ref{alg:init} is often insufficient to produce a $\mathbf{Y}_g$ that is full rank
and thus we require many calls to this subroutine to produce a suitable $\mathbf{Y}_g$.
We further note that in Figure {\ref{fig:speedbyk}}, at large $k$, the runtime appears to grow approximately quadratically by $k$ rather than linearly as predicted by our best-case analysis.  
This is because our current implementation does not fully exploit the sparsity of the tableau matrix; doing so would in fact imposes a performance penalty for small $k$.
We now discuss the cases in which our subroutines fail in greater detail. 

\blockcomment{
\begin{table}
\caption{Runtime and average calls required per solution for Algorithms 1 (Vertex Finding) and 3 (Simplex with Backtracking)}
\label{tab:failure}
\centering
\scalebox{0.85}{
\begin{tabular}{|c|c||c|c||c|c|}
\hline
 \multicolumn{2}{|c||}{} & \multicolumn{2}{c||}{Vertex Finding} &
\multicolumn{2}{c|}{Simplex with Backtracking}\\
\hline
n & k 	& Avg. Calls & Time/Call (s) & Avg. Calls & Time/Call (s) \\
\hline
2 & 8 	& 1.01 & 5.33e-6 & 1.00 & 1.29e-5 \\
3 & 13	& 1.01 & 1.89e-5 & 1.00 & 4.57e-5 \\
4 & 18	& 1.03 & 5.68e-5 & 1.00 & 1.15e-4 \\
5 & 18	& 1.28 & 9.78e-5 & 1.06 & 1.63e-4 \\
6 & 22	& 1.75 & 2.31e-4 & 1.20 & 3.69e-4 \\
8 & 30	& 4.47 & 6.11e-4 & 1.41 & 1.18e-3 \\
\hline
\end{tabular}
}
\end{table}
}

 \begin{table}
\caption{Runtime and average calls required per solution for Algorithms \ref{alg:init} (Vertex Finding) and \ref{alg:main} (Simplex with Backtracking)}
\vspace{2mm}

\label{tab:failure}
\centering
\ra{1.3}
\scalebox{0.76}{
\begin{tabular}{@{}llllllllll@{}}
\toprule
 \multicolumn{2}{c}{} &&& \multicolumn{2}{c}{Vertex Finding} &&&
\multicolumn{2}{c}{Simplex with Backtracking}\\
\cmidrule{5-6} \cmidrule{9-10}
n & k&&& Avg. Calls & Time/Call (s) &&& Avg. Calls & Time/Call (s) \\
\midrule
2 & 8 	&&& 1.01 & 5.33e-6 &&& 1.00 & 1.29e-5 \\
3 & 13	&&& 1.01 & 1.89e-5 &&& 1.00 & 4.57e-5 \\
4 & 18	&&& 1.03 & 5.68e-5 &&& 1.00 & 1.15e-4 \\
5 & 18	&&& 1.28 & 9.78e-5 &&& 1.06 & 1.63e-4 \\
6 & 22	&&& 1.75 & 2.31e-4 &&& 1.20 & 3.69e-4 \\
8 & 30	&&& 4.47 & 6.11e-4 &&& 1.41 & 1.18e-3 \\
\bottomrule
\end{tabular}
}
\end{table}

\begin{figure}
\centering
\begin{tikzpicture}
\begin{semilogyaxis}[
    grid=both,
    width=7cm, height=7cm,
    xmax = 150,
    ymin = 0.00001,
    ymax = 0.02,
    xlabel={$k$ (samples)},
    ylabel={Time per sample (s)},
    legend style={font=\scriptsize,
   		cells={anchor=east},
    	anchor=south east,
    	at={(1.1,0)}}
]

\addplot coordinates
{
(8,1.84e-4)
(10,1.20e-4)
(12,9.18e-5)
(14,6.29e-5)
(16,4.17e-5)
(18,3.61e-5)
(20,3.38e-5)
(22,3.03e-5)
(24,3.00e-5)
(26,2.90e-5)
(28,2.80e-5)
(30,2.86e-5)
(32,2.90e-5)
(34,2.87e-5)
(36,2.83e-5)
(38,2.85e-5)
(40,2.93e-5)
(42,2.98e-5)
(44,3.06e-5)
(46,3.12e-5)
(48,3.25e-5)
(50,3.29e-5)
(52,3.35e-5)
(54,3.43e-5)
(56,3.51e-5)
(58,3.59e-5)
(60,3.99e-5)
(65,4.51e-5)
(70,4.90e-5)
(75,5.45e-5)
(80,5.81e-5)
(85,6.81e-5)
(90,7.45e-5)
(95,7.97e-5)
(100,8.48e-5)
(105,9.97e-5)
(110,1.07e-4)
(115,1.11e-4)
(120,1.18e-4)
(125,1.20e-4)
(130,1.31e-4)
(135,1.34e-4)
(140,1.41e-4)
(145,1.49e-4)
(150,1.62e-4)
(155,1.64e-4)
(160,1.71e-4)
(165,1.77e-4)
(170,1.86e-4)
(175,1.92e-4)
(180,1.99e-4)
(185,2.05e-4)
(190,2.09e-4)
(195,2.16e-4)
(200,2.15e-4)
};

\addplot coordinates
{
(10,4.32e-4)
(12,5.68e-4)
(14,7.08e-4)
(16,6.31e-4)
(18,5.23e-4)
(20,4.02e-4)
(22,2.65e-4)
(24,2.28e-4)
(26,1.90e-4)
(28,1.50e-4)
(30,1.28e-4)
(32,1.13e-4)
(34,1.08e-4)
(36,9.99e-5)
(38,9.24e-5)
(40,8.58e-5)
(42,8.17e-5)
(44,7.94e-5)
(46,7.77e-5)
(48,7.54e-5)
(50,7.59e-5)
(52,7.34e-5)
(54,7.25e-5)
(56,6.79e-5)
(58,6.67e-5)
(60,6.86e-5)
(65,6.93e-5)
(70,7.36e-5)
(75,7.85e-5)
(80,8.37e-5)
(85,8.91e-5)
(90,9.68e-5)
(95,1.05e-4)
(100,1.15e-4)
(105,1.23e-4)
(110,1.33e-4)
(115,1.45e-4)
(120,1.65e-4)
(125,1.89e-4)
(130,1.99e-4)
(135,2.02e-4)
(140,2.02e-4)
(145,2.22e-4)
(150,2.25e-4)
(155,2.50e-4)
(160,2.52e-4)
(165,2.77e-4)
(170,2.79e-4)
(175,2.89e-4)
(180,2.89e-4)
(185,3.14e-4)
(190,3.24e-4)
(195,3.31e-4)
(200,3.37e-4)
};

\addplot coordinates
{
(12,7.87e-3)
(14,7.55e-3)
(15,6.50e-3)
(16,6.69e-3)
(20,5.36e-3)
(25,4.23e-3)
(30,2.89e-3)
(35,2.18e-3)
(40,1.36e-3)
(45,9.14e-4)
(50,5.75e-4)
(55,3.95e-4)
(60,3.07e-4)
(65,2.45e-4)
(70,2.30e-4)
(75,2.21e-4)
(80,2.15e-4)
(85,2.11e-4)
(90,2.06e-4)
(95,1.99e-4)
(100,1.96e-4)
(105,2.08e-4)
(110,2.27e-4)
(115,2.52e-4)
(120,2.73e-4)
(125,2.79e-4)
(130,2.81e-4)
(135,2.91e-4)
(140,3.09e-4)
(145,3.05e-4)
(150,3.04e-4)
(155,3.25e-4)
(160,3.41e-4)
(165,3.41e-4)
(170,3.51e-4)
(175,3.63e-4)
(180,3.81e-4)
(185,4.02e-4)
(190,4.16e-4)
(195,4.28e-4)
(200,4.29e-4)
};

\addplot coordinates
{
(15,1.28e-2)
(20,1.05e-2)
(25,9.32e-3)
(30,8.84e-3)
(35,7.80e-3)
(40,7.63e-3)
(45,7.10e-3)
(50,6.22e-3)
(55,5.52e-3)
(60,5.19e-3)
(65,4.58e-3)
(70,4.10e-3)
(75,3.56e-3)
(80,3.06e-3)
(85,2.65e-3)
(90,2.42e-3)
(95,2.23e-3)
(100,2.02e-3)
(105,1.96e-3)
(110,1.82e-3)
(115,1.65e-3)
(120,1.60e-3)
(125,1.68e-3)
(130,1.72e-3)
(135,1.76e-3)
(140,1.90e-3)
(145,2.07e-3)
(150,2.24e-3)
(155,2.33e-3)
(160,2.46e-3)
(165,2.52e-3)
(170,2.71e-3)
(175,3.11e-3)
(180,3.18e-3)
(185,3.25e-3)
(190,3.35e-3)
(195,3.39e-3)
(200,3.49e-3)
};

\legend{$n=6$, $n=8$, $n=10$, $n=12$}
\end{semilogyaxis}
\end{tikzpicture}
\caption{The average time to solution of our algorithm normalized by $k$, the
number of samples. For small values of $k$, the runtime is dominated by failures of Algorithm 1.  For large $k$, the runtime is dominated by the vertex hopping process.}
\label{fig:speedbyk}
\end{figure}
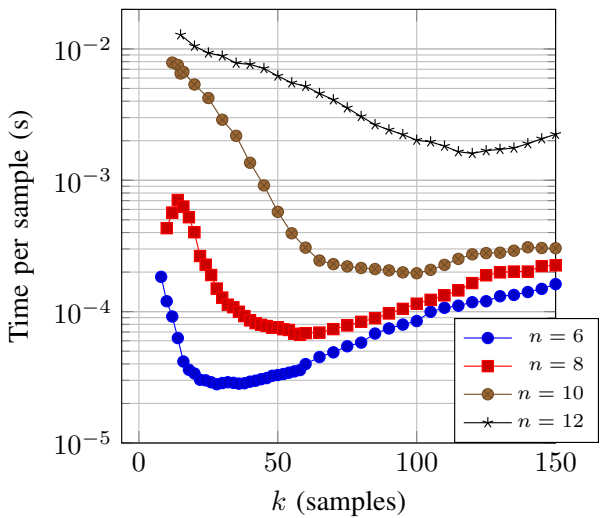

\begin{figure}
\centering
\begin{tikzpicture}
\begin{semilogyaxis}[
    grid=both,
    width=7cm, height=7cm,
    xmax = 20,
    ymin = 3e-6,
    ymax = 1,
    xlabel={$k$ (samples)},
    ylabel={Pr. $\left[(\mathbf{UY})_{ij} \in \mathcal{S}\right]$},
    legend style={font=\scriptsize,
    	cells={anchor=west},
    	anchor=south west,
        at={(0,0)}}
]

\addplot coordinates
{
(5, 0.955831)
(7, 0.962727)
(9, 0.971772)
(11, 0.977612)
(13, 0.983557)
(15,0.987091)
(17, 0.990305)
(19, 0.992425)
};

\addplot coordinates
{
(5, 4.397355e-02)
(7, 3.710139e-02)
(9, 2.817246e-02)
(11, 2.224032e-02)
(13, 1.634076e-02)
(15, 1.289047e-02)
(17, 9.647930e-03)
(19, 7.456711e-03)
};

\addplot[mark=triangle*,
mark options={solid,mark size=3}] coordinates
{
(5, 1.953712e-04)
(7, 1.716002e-04)
(9, 1.123e-04)
(11, 1.037e-04)
(13, 6.716e-05)
(15, 5.6326e-05)
(17, 5.709179e-05)
(19, 6.119e-05)
};

\legend{
	$\mathcal{S} = \{-1,1\}$ \\
	$\mathcal{S} = \{0\}$ \\
	$\mathcal{S} = (-1, 1) \backslash \{0\}$ \\
}
\end{semilogyaxis}
\end{tikzpicture}
\caption{Empirical distribution of the entries of $\mathbf{UY}$, which are in either $\{\pm 1\}, \{0\}$, or $(-1,+1)\backslash \{0\}$, at the output of Algorithm \ref{alg:init} with no noise for $n=4$ and various values of $k$.  We can see that as $k$ grows, the fraction of zero-valued entries decreases, while the fraction of entries in $(-1,+1)\backslash \{0\}$ remains roughly constant.  Data averaged over 100\,000 trials and over all $(i,j) \in [1,\ldots,n]\times[1,\ldots,k]$.
\vspace{4mm}}
\label{fig:dynamic_empirical}
\end{figure}
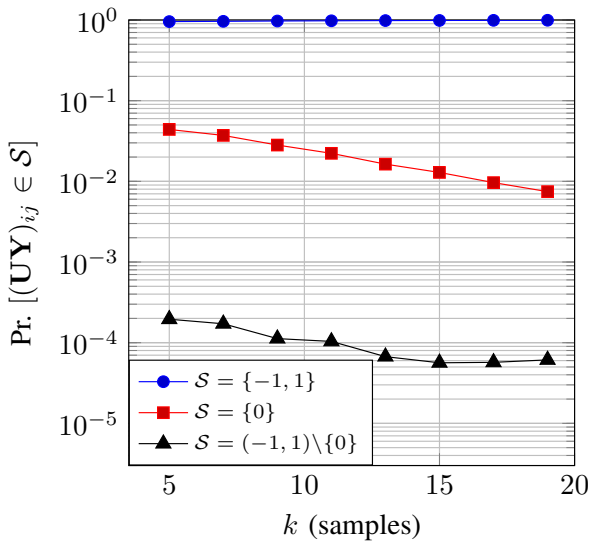
\subsubsection{Insufficient `Good' Columns}
In order to form a simplex tableau, we must find at least $n$ linearly
independent columns that are contained in $\{-1,+1\}^n$.  
When $n=k$ and $\mathbf{X}$ is maximal, we are guaranteed to find such a
result.
However, for $k>n$, this is no longer guaranteed.  
We find that for small $n$, as well as for all $n$ when $k$ is large, we obtain such a
result with high probability.  
However, we find that for large $n$ and small $k$, the probability that Algorithm \ref{alg:init} fails to produce a suitable output becomes non-negligible.  

Our empirical results suggest that for $n \leq 5$, Algorithm \ref{alg:init} is certainly sufficient to initialize Algorithm \ref{alg:main}.  It is an active area of research to determine whether or not
alternative initialization methods may help improve the average
runtime of our approach for larger $n$.

\subsubsection{Trap Cases}
For $n=6$ and $n=8$, it is possible to obtain a BFS that lies in a
component of the graph of vertices that does not contain a global optimum.
We refer to such cases as \emph{trap cases}.
In these cases the solver will enumerate the entire component and then exit
without returning a solution to the blind decoding problem.
For both $n=6$ and $n=8$, empirically, we see that the odds of finding such trap
cases decreases as $k$ increases. 
Roughly, this occurs because either portions of these components, or the entire
component, becomes infeasible.
 
We find that when the solver is presented a BFS that lies on a graph component
with a global optimum, for $n \leq 8$, the solver typically only requires
a very small number of hops (rarely more than 4) in order to find this solution.  
However, we note that as $n$ increases, the number of vertices of the feasible region grows exponentially fast; 
and we find that, for $n=8$, a trapped subspace may contain tens of thousands of non-singular vertices.
The default behavior of the algorithm is to enumerate this entire subspace in attempt to find a global optimum.  
As a simple heuristic to detect a trap cases, we simply limit the maximum number of vertices that the solver may visit in a single attempt.
If the solver exceeds this limit, we assume that the solver has encountered a trap case, and we proceed by returning to Algorithm \ref{alg:init} and finding a new BFS.
Empirically, we find that setting this limit to $2nk$ is sufficiently large to almost always avoid falsely detecting a trap case.  
It is an open question to determine if there exist more
intelligent methods to recover from a trap case rather than restarting the
solver from scratch.

\subsubsection{False Trap Cases}
For all values of $n$, we find that there is a small probability that Algorithm \ref{alg:main} will incorrectly terminate after enumerating the entire state space, thus presenting what appears to be a trap case.
These `false traps' are caused by numerical instability; the solver may deem a
neighboring vertex to be infeasible because the values of one or more entries of
$\mathbf{UY}$ exceeds $\pm 1$ by more than a predefined floating point
threshold.

For $n < 5$, the odds of encountering a false trap is less that one per one thousand with this probability increasing slightly as $n$ grows further.
When such an error occurs, we typically find that the channel gain matrix is poorly conditioned ($\kappa > ~10^5$).  
Surprisingly, even with such poorly conditioned channels, we can often recover from this type of error by simply obtaining a new BFS and trying again.
This behavior indeed seems difficult to avoid for extremely ill-conditioned channels. Because this behavior occurs so infrequently and only in poorly conditioned channels, we believe that simply restarting is a sufficient  solution and we do not need to consider further optimizations.

\section{Robust Decoding}
\label{sub:noise}
\subsection{Algorithm}
So far, we have only considered solving the blind decoding problem in the limit of high SNR.  We now turn our attention to how to robustly solve the blind decoding problem in the presence of AWGN.  We make no claims that the technique provided here is optimal in terms of BER performance; it almost certainly is not.  However, our technique works well empirically and is computationally efficient.  
Our treatment of this algorithm is entirely empirical; an analytic treatment of the performance of our algorithm in the presence of noise is a topic of ongoing research.

In order to understand how to adapt to noise, we consider the output of Algorithm \ref{alg:init} in the noiseless case in more detail.
Empirically, for any $(n,k)$, when Algorithm \ref{alg:init} terminates, the vast majority of entries of $\mathbf{UY}$ are contained in the set $\{-1,+1\}$.  A small proportion, on the order of 1\%, of entries are 0-valued and an even smaller proportion (roughly two out of 10\,000) appear uniformly distributed within $(-1,+1)\backslash 0$.  Empirical results describing the distribution of the entries of $\mathbf{UY}$ for $n=4$ and various values of $k$ is shown in Figure \ref{fig:dynamic_empirical}.

When AWGN is present, the behavior of Algorithm \ref{alg:init} changes drastically because the columns of $\mathbf{Y}$ no longer have such strong linear dependence. With near certainty, the output of Algorithm \ref{alg:init} will have only $n^2$ entries that are exactly $\{-1,+1\}$, and, as $k$ grows, this means that $\mathbf{Y}_g$ will almost never be full rank and we will be unable to proceed with Algorithm \ref{alg:main}.


To address this issue, we add an additional rounding step to Algorithm \ref{alg:init}.  If entries of $\mathbf{UY}$ are within some tolerance, \change{$\epsilon$}, of 
$\{-1,0,1\}$, we simply adjust $\mathbf{Y}$ to effectively round off the corresponding entry of $\mathbf{UY}$.  Explicitly, we compute each a matrix \change{$\Sigma$} where the $ij$th entry is given by
\begin{align}
\Sigma_{ij} =
\begin{cases}
(\mathbf{UY})_{ij} + 1 & |(\mathbf{UY})_{ij} + 1| < \epsilon \\
(\mathbf{UY})_{ij} - 1 & |(\mathbf{UY})_{ij} - 1| < \epsilon\\
(\mathbf{UY}) & |(\mathbf{UY})_{ij}| < \epsilon \\
0 & \text{otherwise.}
\end{cases} 
\label{eq:center}
\end{align}
$\mathbf{Y}$ is then updated by computing $\mathbf{Y} = \mathbf{Y} - \mathbf{U}^{-1} \Sigma$.  Performing this rounding step only once is typically not sufficient to find a value of $\mathbf{Y}_g$ that is sufficient to construct a tableau.  Indeed, rounding only once is often insufficient when $\mathbf{X}$ has columns that are identical up to a sign.  In these cases, $\mathbf{Y}_g$ will be deficient in rank by the number of identical columns of $\mathbf{X}$.  More independent columns of $\mathbf{Y}_g$ may be found at this point by returning to the main procedure of Algorithm \ref{alg:init} after rounding.  

The full rounding process is presented in Algorithm \ref{alg:robust}. 
In practice, we find that the \texttt{for} loop contained in this algorithm will often terminate after only one or two iterations.  
Repeating the rounding process beyond $n$ times will have no effect;
after completing $n$ iterations of the process we are guaranteed that the set of active constraints will be full rank.

\begin{algorithm}
 \caption{Vertex Finding in the presence of AWGN}
 \label{alg:robust}
 \begin{algorithmic}[1]
 \renewcommand{\algorithmicrequire}{\textbf{Input:}}
 \renewcommand{\algorithmicensure}{\textbf{Output:}}
 \REQUIRE $\mathbf{Y}$
 \ENSURE  $\hat{\mathbf{Y}}$, $\mathbf{U}$ which is a BFS.
 \STATE Draw $\mathbf{U}_0$ at random as usual
 \STATE $\mathbf{Y}_0 = \mathbf{Y}$
 \FOR {$i = 0, i < n, i++$}
	\STATE $\mathbf{U}_{i+1} = 
		 	\mathtt{Algorithm\;\ref{alg:init}}
            \left(\mathbf{U}_i,\mathbf{Y}_i\right)$
    \IF {$\| \mathbf{U}_{i+1} - \mathbf{U}_i \|_\infty < \epsilon$}
    	\STATE break
    \ENDIF
    \STATE Compute $\Sigma$ as in eq. \pref{eq:center}.
    \STATE $\mathbf{Y}_{i+1} = \mathbf{Y}_i - 			\mathbf{U}_{i+1}^{-1} \Sigma$
 \ENDFOR
 \RETURN $\mathbf{U}_{i}$, $\mathbf{Y}_{i}$.
 \end{algorithmic} 
 \end{algorithm}

\subsection{Choosing $\epsilon$ and BER Performance}
\label{sub:ber}
We now turn our attention to discussing the BER performance of the vertex hopping algorithm using Algorithm {\ref{alg:robust}}.
We first consider the behavior of this algorithm in the high SNR limit.
From Figure \ref{fig:dynamic_empirical}, we see that roughly 2 out of $10\,000$ of the entries are distributed in the interval $(-1, 1)\backslash 0$.  When these entries are within $\epsilon$ of $\pm 1$ then Algorithm \ref{alg:robust} will incorrectly force these entries to $\pm 1$.  This may result in recovering an estimate of $\mathbf{A}^{-1}$ that is not equivalent up to an ATM.  This behavior explains the noise enhancement observed at high SNR in Figure {\ref{fig:awgn}}.
\change{ We note that one impractical solution to avoid this effect would be to run the gradient descent algorithm of {\cite{dean2018blind}} using the output of the vertex hopping algorithm as a starting point. However, this would greatly increase the runtime of the algorithm.  It is an open problem whether or not this phenomena can be avoided without drastically increasing the runtime of the algorithm.}


In the presence of AWGN, we observed that the BER is effectively constant for a given value of $\epsilon$.  
Instead, with a fixed value of $\epsilon$, as SNR decreases, the odds that the solver will complete decays.
This is because, if $\epsilon$ is small compared to the noise variance, the rounding procedure described above is unlikely to have any effect. 
As a result, it is unlikely that Algorithm \ref{alg:robust} will yield a value of $\mathbf{Y}$ with a full rank $\mathbf{Y}_g$.
The odds that Algorithm \ref{alg:robust} produces a suitable output as a function of SNR for various values of $\epsilon$ is presented \change{ on the left-hand side of Figure } \ref{fig:completed}. 

Unlike the trap cases discussed in Section \ref{sub:trap}, the odds of Algorithm \ref{alg:robust} succeeding is largely dependent on the input $\mathbf{Y}$ and not the initial choice of $\mathbf{U}_0$.
In other words, if the solver fails at low SNR, it is unlikely to succeed if it is run again with a different choice of $\mathbf{U}_0$.
In practice, in the low SNR regime, we run several iterations with different choices of $\mathbf{U}_0$ to account for failures due to trap cases, but if these attempts fail we must either declare an erasure or raise the value of $\epsilon$.  
This behavior allows us to choose the value of $\epsilon$ based requisite bit- and block-error rates. \change{The right-hand side of 
Figure} \ref{fig:completed} shows the average BER as a function of $\epsilon$.  These results are averaged over all SNR values ranging from 10dB to 30dB.  

\begin{figure*}
\centering
\hfill
\begin{tikzpicture}
\begin{axis}[
    grid=both,
    width=7cm, height=7cm,
    ymin = 0,
    ymax = 1,
    xlabel={SNR (dB)},
    ylabel={Success Probability},
    legend style={font=\scriptsize,
    	cells={anchor=west},
    	anchor=south east,
        at={(1.05,0)}}
]

\addplot coordinates
{
(30,0.969)
(27,0.895)
(24,0.662)
(21,0.404)
(20,0.230)
(17,3.1e-2)
};

\addplot coordinates
{
(30, 0.967)
(27, 0.957)
(24, 0.912)
(21, 0.719)
(20, 0.639)
(17, 0.309)
(14, 1.95e-2)
};

\addplot coordinates
{
(30, 0.996)
(27, 0.989)
(24, 0.967)
(21, 0.928)
(20, 0.862)
(17, 0.647)
(14, 0.340)
(11, 3.88e-2)
(10, 1.13e-2)
};

\addplot coordinates
{
(30, 0.998)
(27, 0.985)
(24, 0.989)
(21, 0.934)
(20, 0.939)
(17, 0.866)
(14, 0.518)
(11, 0.252)
(10, 0.106)
};

\addplot coordinates
{
(30, 0.994)
(27, 0.995)
(24, 0.999)
(21, 0.971)
(20, 0.965)
(17, 0.925)
(14, 0.672)
(11, 0.330)
(10, 0.241)
};

\addplot coordinates
{
(30, 0.999)
(27, 0.998)
(24, 0.996)
(21, 0.998)
(20, 0.990)
(17, 0.951)
(14, 0.867)
(11, 0.603)
(10, 0.530)
};

\addplot coordinates
{
(17, 0.969)
(14, 0.923)
(11, 0.821)
(10, 0.706)
};

\addplot coordinates
{
(14, 0.971)
(11, 0.956)
(10, 0.945)
};

\legend{$\epsilon=0.025$,
		$\epsilon=0.05$, $\epsilon=0.1$, 
		$\epsilon = 0.15$,$\epsilon = 0.20$,
        $\epsilon = 0.25$, $\epsilon = 0.33$,
        $\epsilon = 0.5$}
\end{axis}
\end{tikzpicture}
\hfill
\begin{tikzpicture}
\begin{semilogyaxis}[
    grid=both,
    width=7cm, height=7cm,
    ymax = 1,
    xlabel={$\epsilon$},
    ylabel={BER},
    legend style={font=\scriptsize,
    	cells={anchor=west},
    	anchor=south east,
        at={(0.95,0.05)}}
]
\addplot coordinates
{
(0.025,1.79e-4)
(0.05,2.45e-4)
(0.10,4.21e-4)
(0.15,1.07e-3)
(0.20,3.28e-3)
(0.25,7.77e-3)
(0.33,1.79e-2)
(0.50,6.78e-2)
};

\end{semilogyaxis}
\end{tikzpicture}
\hspace{5mm}
\caption{On the left, the probability of the vertex hopping method completing (that is, returning a value of $\hat{\mathbf{X}}$ that may or may not be correct up to an ATM) for different values of $\epsilon$. On the right, the BER of the solver is averaged across SNR values ranging from 10dB to 30dB and plotted against $\epsilon$.  The BER of our approach is almost entirely a function of the value of $\epsilon$ and not SNR.}
\label{fig:completed}
\end{figure*}
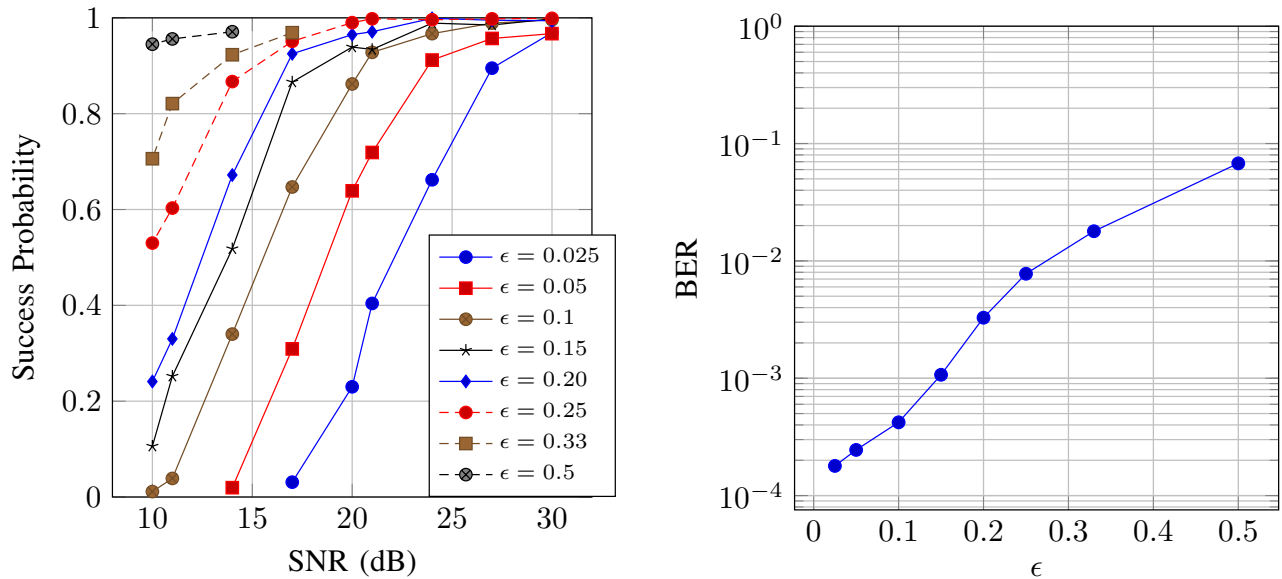
Finally, we note that, in addition to the noise variance, the condition number of the channel gain matrix also plays a large role in both the success probability depicted in Figure \ref{fig:completed} and the BER performance.
\change{Indeed we find a very high correlation between the failure rate of our algorithm and the condition number of the channel.  
A more optimal strategy would likely entail choosing} $\epsilon$, for example, based in part on the singular values of $\mathbf{Y}$, rather than exclusively on the SNR. 
We did not explore how to exploit this relationship as computing the SVD of $\mathbf{Y}$ adds significant complexity at the receiver.  
Further, if the distribution of the channel gain matrix is such that the variance of its condition number is small over each block then such an approach would be of little benefit.

\section{Conclusion}
\label{sec:conclusion}
In this work, we have presented a `vertex hopping' algorithm that can efficiently and reliably perform blind decoding of BPSK MIMO symbols in the presence of AWGN.  
We present the problem of decoding as a non-linear mixed-integer program and leverage techniques from solving linear programming to solve it.
Our vertex hopping algorithm consists of a two-step process; we must first find a suitable vertex of the feasible region that describes the integer program.  
Once this vertex is found we hop between vertices of the feasible region in an attempt to search for a global optimum.

Empirically, we show that this technique is both efficient and reliable for MIMO systems as large as $n=8$.
Our technique still works beyond $n=8$, but the underlying non-convex optimization problem appears to become computationally difficult as $n$ grows. 
For small $n$, the vertex hopping algorithm is both efficient and practical in terms of sample and computational complexity.
At low SNR, our technique performs comparably to zero-forcing decoding and at all SNRs outperforms ML decoding with as little as 1\% CSI error.

The algorithm presented in this work motivates a suite of future research topics.  
Many of these topics are related to understanding and improving the performance of our algorithm in a variety of operating conditions. 
Examples include: developing optimizations to exploit the structure present in complex channel gain matrices, studying the performance of our algorithm under higher modulation orders, and considering rectangular channels.
Further, more research is warranted at the wireless-systems level to understand how to best use our algorithm to build high-rate, reliable MIMO systems that operate in environments with rapidly changing CSI.


\appendices

\section{Stopping Criteria}
\label{apx:stop}
In this appendix, we prove that the criterion given in Section \ref{sub:stop} are indeed necessary and sufficient to determine whether a value of $\mathbf{U}$ is at a global optimum.  
For simplicity of exposition, we simply consider the case $k=n$; the arguments contained in this section indeed still hold when $k>n$ as we can still compute the objective function of all $n^2$ neighboring vertices regardless of whether or not they are feasible.

\begin{theorem}
\label{theorem:stop} 
For $n \leq 6$, a global optimum can be detected based on the value of the objective function obtained on each of its $n^2$ neighbors.
\end{theorem}

This theorem follows trivially for $n \leq 5$ due to the fact that the only optima contained on the vertices of the feasible region are in fact global in these cases. 
This is proven in \cite{dean2018blind} and \cite{perlstein2018n5case}.
Before proving this theorem for $n=6$, we formalize the discussion in Section \ref{sub:stop} regarding the optima of \pref{eq:obj3}--\pref{eq:const3}.
Define the set $\mathcal{D} = \{|\det \mathbf{X}| : \mathbf{X} \in \{-1,+1\}^{n \times n} \}$, which in \cite{orrick_spectrum} is referred to as the \emph{spectrum} of possible determinants.  Since, in the noiseless case, $\det \mathbf{UY} = \det \mathbf{UAX}$, and $| \det \mathbf{X}| \in \mathcal{D}$, this implies the following claim:
\begin{claim}
\label{claim:spectrum}
The value of $|\det \mathbf{U}| = D |\det \mathbf{A}^{-1}|$ for some $D \in \mathcal{D}$.
\end{claim}

We say that a $\pm 1$-valued matrix is \emph{maximal} if it obtains the maximum determinant amongst all matrices constrained to $\{-1,+1\}^{n \times n}$.  
Suppose there are $N$ distinct maximal matrices, then we denote these matrices as $\mathbf{X}_1, \ldots, \mathbf{X}_N$.

\begin{defn}
Equivalence of matrices. Two matrices $\mathbf{X}_1$, $\mathbf{X}_2$ are weakly
equivalent if, for some $\mathbf{T}_1, \mathbf{T}_2 \in \mathcal{T}$,
$\mathbf{X_1} = \mathbf{T}_1 \mathbf{X}_2 \mathbf{T}_2$.  If $\mathbf{X}_1 =
\mathbf{T}_1 \mathbf{X}_2$, then $\mathbf{X}_1, \mathbf{X_2}$ are strongly
equivalent.
\end{defn}

Strong equivalence is the same as being equivalent up to an ATM, which only allows permutation and negation of columns. 
All solutions to the blind decoding problem are strongly equivalent.
However, matrices that are weakly equivalent to $\mathbf{X}$ may have both permuted rows and columns.  Such matrices are not necessarily solutions to the blind decoding problem and may correspond to spurious optima.

We now consider the poset formed by taking a single vertex and its $n^2$ neighbors and imposing an ordering on the graph.  Specifically, the poset $(\mathcal{X}, \leq)$ formed by taking the set, $\mathcal{X}$, containing
$\mathbf{UY}$ and its $n^2$ neighbors, and imposing the ordering where  $\mathbf{A} \leq \mathbf{B}$ implies 
$| \det \mathbf{A} | \leq |  \det \mathbf{B} |$.

\begin{lemma}
\label{lemma:isomorphic}
Suppose $n$ is such that all maximal $\pm 1$-valued matrices are weakly equivalent and that there are $N$ maximal matrices.  
Then all posets  $(\mathcal{X}_1, \leq), \ldots, (\mathcal{X}_N, \leq)$ are isomorphic. 
\end{lemma}

\begin{proof}
Consider the mapping 
$\phi_{k,l} : \mathbb{R}^{n \times n} \rightarrow \mathbb{R}^{n \times n}$ given by 
$\mathbf{X}_i \mapsto \mathbf{T}_k \mathbf{X}_i \mathbf{T}_l$, for some $\mathbf{T}_k, \mathbf{T}_l \in \mathcal{T}$.  
This mapping is one-to-one since all elements of $\mathcal{T}$ are square and full rank. 
This mapping also preserves the ordering of the poset $\left(\mathcal{X}_i, \leq \right)$.
This is because permuting rows or columns of a matrix can only change the sign of its determinant, and likewise for negating its rows or columns; hence, for any $\mathbf{X}, \mathbf{T}_k, 
\mathbf{T}_l$, we have 
$| \det \mathbf{X} | = | \det \mathbf{T}_k \mathbf{X} \mathbf{T}_l |$.  
If $\mathbf{X}_j = \phi_{k,l}(\mathbf{X}_i)$, then this implies that the poset $(\mathcal{X}_i,\leq)$ is isomorphic to $(\mathcal{X}_j,\leq)$. 
Since $n$ is such that there is only one weak equivalence class of maximal determinant matrices, this implies that for any $\mathbf{X}_i$ and $\mathbf{X}_j$, there exists a $k,l$ such that $\mathbf{X}_j = \phi_{k,l}(\mathbf{X}_i)$.
Thus, the posets $(\mathcal{X}_i,\leq)$, $(\mathcal{X}_j,\leq)$ are isomorphic for all $i,j$.
\end{proof}
Up to $n=10$, and for several values beyond $n>10$, all maximal matrices are weakly equivalent (see \cite{smith1989studies}); 
Lemma \ref{lemma:isomorphic} holds for these cases.   
Beginning at $n=6$, local optima exist in \refprog.  However, the following lemma shows that the poset obtained at any global optima in this case is uniquely identifiable.
\begin{lemma}
For $n=6$, the poset $(\mathcal{X}, \leq)$ obtained at a global optimum is uniquely identifiable.
\end{lemma}
\begin{proof}
For $n=6$, we have $\mathcal{D} = \{0, 32, 64, 128, 160\}$.  
By Lemma \ref{lemma:isomorphic}, and the fact that all maximal matrices at $n=6$ are weakly equivalent, we can consider the neighbor pattern of any one maximal vertex.  We now require the following claim, which may be verified by inspecting all 36 neighbors of any maximal $n=6$ matrix.
\begin{claim}
\label{claim:n6_maximal}
Any maximal matrix at $n=6$ has only neighbors with determinants $\pm 128, \pm 96$, and $\pm 64$.  
\end{claim}

Given any particular $\mathbf{UY}$, we can determine the value of the objective function obtained by using the procedure outlined in Section \ref{sub:pivot}.  By Claim \ref{claim:n6_maximal}, if a matrix is maximal the objective functions at its neighbors must take on three distinct non-zero values.
Thus, if the neighbors of $\mathbf{UY}$ have three distinct possible non-zero values, this immediately implies that the value of \pref{eq:obj3} obtained at $\mathbf{UY}$ must be either 160 or 128.  

We can distinguish between these two cases by simply considering the relative change between the value of the objective function at $\mathbf{UY}$ and one of its neighbors that is closest in value. 
Call such a candidate vertex $\mathbf{U'Y}$.  
If $\mathbf{UY}$ is maximal, by Claim \ref{claim:spectrum}, we must have $|\det \mathbf{U}' / \det \mathbf{U}| = 4/5$.   
This is because 160 the only element of $\mathcal{D}$ that is divisible by 5. This process uniquely determines that $\mathbf{UY}$ is maximal; if $\mathbf{UY}$ is not maximal than $|\det \mathbf{U}' / \det \mathbf{U}| < 4/5$.
\end{proof}

This completes the proof of Theorem \ref{theorem:stop}.
Proving similar results for larger values of $n$ appears difficult as $n$ grows.  Not only does the number of $\pm 1$ matrices grow exponentially in $n$, but also size of the set $\mathcal{D}$ grows rapidly.  However, we do conjecture that the following criteria is sufficient for detecting global optima at $n=8$.
\begin{conj}
For $n=8$, maximal vertices are the only vertices such that for all 64 neighbors, the objective function decreases by a factor of 0.25.
\end{conj}
It is not hard to see that this condition is necessary. For $n=8$, there is only one equivalence class of strongly equivalent maximal matrices, and such matrices have a determinant of $\pm 4096$. 
Further, for any maximal matrix, changing any single entry results in a matrix with determinant $\pm 3072$.\footnote{This can be verified by checking any maximal matrix.  An example of a maximal matrix at $n=8$ can be obtained by taking the Kronecker product of three 2-dimensional Hadamard matrices.  Lemma \ref{lemma:isomorphic} ensures that all maximal matrices will have the same property.}  
Empirically, we have not found any non-maximal matrices where a similar property holds.  However, we defer on a rigorous proof of this conjecture.
We state similar conjectures for the cases $n=10$ and $n=12$.
\begin{conj}
For $n=10$, a vertex is maximal if and only if 20 of its 100 neighbors decrease the objective function by a fraction of 1/3 and 80 of its 100 neighbors decrease the objective function by a fraction of 1/6.
\end{conj}
\begin{conj}
For $n=12$, a vertex is maximal if and only if all 144 of its neighbors decrease the value of the objective function by a fraction of 1/6.
\end{conj}
\bibliographystyle{ieeetr}
\bibliography{IEEEabrv,trdean}

\end{document}